\documentclass[twoside, 12pt]{article}
\usepackage{pgfplots,comment,verbatim}
\usepackage[a4paper,top=1.18in, bottom=1.18in, left=1.2in, right=1.2in,bindingoffset=0mm]{geometry}

\usepackage{aditya_ord}
\tikzset{
    >=latex,
    pil/.style={
            draw,
      <-, 
      decorate,
      decoration={snake,,amplitude=.02cm, pre length=.2cm,post length=.2cm,}
              }}
\hypersetup{colorlinks=true,linkcolor=red, citecolor= blue, urlcolor=darkred}

\definecolor{BluishGreen}{RGB}{0,158,115}

\title{{\bf \sc {{\color{darkblue}Similarity of Information in Games}}}\footnote{This research was supported in part by National Science Foundation award 2417694. Kuvalekar thanks the Indian School of Business for their hospitality. We are grateful to Nageeb Ali, Dirk Bergemann, Geoffroy De Clippel, Duarte Goncalves, Nima Haghpanah, Johannes H\"orner, Elliot Lipnowski, Michael Mandler, Meg Meyer, Dilip Mookherjee, Stephen Morris, Alessandro Pavan, Ludvig Sinander, Bruno Strulovici, and  conference audiences and seminar participants at Bonn University, Brown University, Carnegie Mellon University, Cornell University, Indian Statistical Institute Delhi, National University of Singapore,  Ohio State University, Oxford University, Paris School of Economics, Royal Holloway University, Toulouse School of Economics, University of Rochester, University of Surrey, Duke \& UNC, and Penn State for their helpful comments and suggestions. Basak: Indiana University, Kelley School of Business, email: dbasak@iu.edu; Deb: New York University, email: joyee.deb@nyu.edu, Kuvalekar: University of Essex, email:  a.kuvalekar@essex.ac.uk.} \\}

\author{                   
\begin{minipage}{0.3\textwidth}\centering 
Deepal Basak\\ \centering \it \small Indiana University
\end{minipage}                  
\begin{minipage}{0.3\textwidth}\centering 
Joyee Deb   \\ \centering \it \small New York University
\end{minipage}                  
\begin{minipage}{0.3\textwidth}\centering 
Aditya Kuvalekar%
\\ \centering \it \small University of Essex
\end{minipage}   
}

\date{\today}

\begin{document}
\maketitle
\bigskip
\begin{center}
\end{center}

\vspace{-0.4cm}

\begin{abstract}

Algorithmic content targeting homogenizes information, with implications for strategic interactions. For example, this increased homogenization was arguably responsible for the run on the Silicon Valley Bank. We argue that existing measures of similarity are inappropriate for studying games—especially coordination games—because they do not discipline agents' conditional beliefs. We propose a class of stochastic orders, Concentration Along the Diagonal (CAD), built on agents' conditional beliefs. In canonical binary-action coordination games, greater CAD-similarity is both necessary and sufficient for strategic similarity—agents adopt the same strategy. We further demonstrate CAD's applicability in congestion games, collective action, and second-price auctions. 

\bigskip

\vspace{0.2cm}

\end{abstract}
\section{Introduction}
In recent times, algorithmic targeting of content has resulted in increased homogenization of information among like-minded individuals: People with similar demographic traits or preferences are exposed to the same content.\footnote{A large literature documents the decrease in diversity of content consumption across users as a result of recommendation algorithms. See, for example, \cite{nechushtai2024more, anwar2024filter, nguyen2014exploring,chaney2018algorithmic,aridor2020deconstructing} for evidence of homogenization in various contexts from online purchases to news consumption.} This homogenization can affect agent behavior in strategic situations. For instance, both empirical research and news media have highlighted how access to the same information on Twitter encouraged many people to withdraw deposits, precipitating a run on Silicon Valley Bank in 2023.\footnote{See 
\href{https://www.imf.org/en/Publications/fandd/issues/2024/03/Containing-Technology-driven-bank-runs-Krogstrup-Sangill-Sicard}{1}, 
\href{https://fortune.com/2023/03/30/how-did-silicon-valley-bank-fail-whatsapp-chats-messages-twitter}{2}, 
\href{https://www.theguardian.com/business/2023/mar/16/the-first-twitter-fuelled-bank-run-how-social-media-compounded-svbs-collapse}{3}, or  \cite{cookson2023social,gam2023does} on bank runs.} The narrative is that people accessed the same information about the health of the bank and reacted to this information in the same way, causing a run which may not have been possible without the homogenization of information. 
Motivated by this discourse, we seek to systematically investigate how increased information similarity across agents affects their strategic interactions. 

We present a simple bank-run example that motivates the central question of our paper. This example demonstrates that increasing information similarity across agents in some natural ways---like increasing correlation---does not increase the expected number of people that would run on the bank! Motivated by this observation, we ask: What type of information homogenization encourages agents to follow the same strategy as others in the face of incomplete information in coordination games such as bank runs? Our main contribution is to answer this question by proposing a new class of stochastic orders—Concentration Along the Diagonal (CAD)—to compare similarity of information.

\subsection*{\small{Correlation and Coordination: A puzzle}} Two agents each decide whether to stay (S) or run (R) on a bank. The payoff from running is $0$. The payoff from staying is $\theta$ if the other player stays, and $\theta-1$ if the other runs, where $\theta$ is an underlying state, say the health of the bank, which is unknown:~$\theta \in \left\{\frac{3}{2},\frac{1}{2}, -\frac{1}{2}\right\}$. 
\begin{center}
\begin{tabular}{|c|c|c|}
\hline 
& Stay & Run \\
\hline
Stay & $\theta, \theta$ & $\theta-1, 0$ \\ 
\hline 
Run & $0, \theta-1$ & $0, 0$\\
\hline
\end{tabular}
\end{center}
Notice that staying is dominant if $\state = \frac32$ and running is dominant if $\state = -\frac12$. The state~$\theta$ is drawn from the prior~$P(\theta=\frac{3}{2})=P(\theta=-\frac{1}{2})=\epsilon$ and $P(\theta=\frac{1}{2})=1-2\epsilon$. Each agent $i$ sees a private signal~$s_i$ about the realized~$\theta$. The signals~ $(s_1,s_2)$ in state~$\theta$ are drawn from some joint distribution~$\distq^\theta$ over $\left\{-\frac{1}{2}, \frac{1}{2},\frac{3}{2}\right\}\times \left\{-\frac{1}{2}, \frac{1}{2},\frac{3}{2}\right\}$. Assume that signals are sufficiently accurate so that $R$ is strictly dominated for an agent with signal $s_i=\frac{3}{2}$ and $S$ is strictly dominated for an agent with signal $s_i=-\frac{1}{2}$.\footnote{A sufficient condition is  $P(s_i=\theta|\theta)=p$ and $P(s_i=\theta'|\theta)=\frac{1-p}{2}$ for $\theta'\neq \theta$ with $p>\frac{3-3\epsilon}{3-\epsilon}$.}

Consider the largest run on the bank possible in any symmetric pure strategy equilibrium. This would involve each  agent~$i$ staying if and only if she gets a signal~$s_i=\frac{3}{2}$. Essentially, the existence of such an equilibrium depends only on whether a player with a signal $\frac12$ wishes to run or not. This, in turn, solely depends on how likely it is, according to this player, that the other player receives a signal $\frac12$ or $\frac{-1}{2}$. 

Suppose we started with an information structure where such an equilibrium exists and make 
players' signals more correlated. 
Intuitively, we should expect more correlated information to cause more people to run on the bank, because, conditional on their own private signal, each agent now thinks that it is more likely that others have also seen a similar signal and will run on the bank. For instance, at the extreme, if all the agents see exactly the same information, then for any observed signal, if all others are willing to run on the bank, and running undominated, then any individual will also run on the bank. Surprisingly, it is not true that more correlated information makes  larger bank runs possible.

The figure below shows two joint distributions~$\distq^{\theta=\frac{1}{2}}$ and $\dist^{\theta=\frac{1}{2}}$ of signals in state~$\theta=\frac{1}{2}$,  
where $\dist^{\theta=\frac{1}{2}}$ has a higher correlation.\footnote{For simplicity, assume that the signals are conditionally independent when $\theta\in\{-\frac{1}{2},\frac{3}{2}\}$.} For example, the $-\alpha$ and $+\alpha$ in the top row mean that the probability mass on $(-\frac{1}{2},\frac{3}{2})$ under $\dist^{\theta=\frac{1}{2}}$ is $\alpha$ less than under $\distq^{\theta=\frac{1}{2}}$, and the probability mass on $(\frac{1}{2},\frac{3}{2})$ under $\dist^{\theta=\frac{1}{2}}$ is $\alpha$ more than under $\distq^{\theta=\frac{1}{2}}$. 
 It is easy to verify that $\dist$ not only has higher correlation than $\distq$, but also is  greater than $\distq$ in the supermodular order, positive quadrant dependence order or concordance orders, for any positive and feasible $\alpha$.
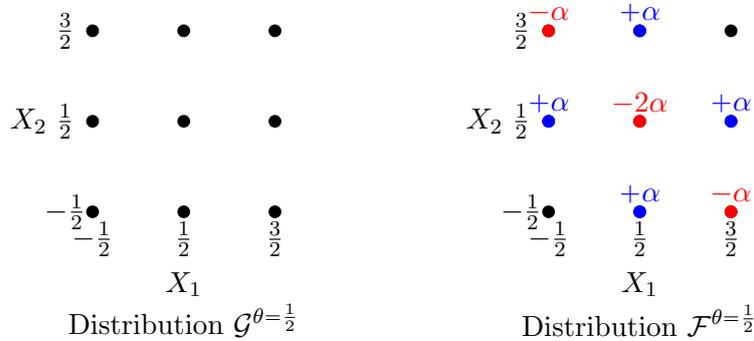
\begin{figure}[h!]
\centering
\begin{tikzpicture}[scale=1.2]
    \foreach \x in {-0.5,0.5,1.5} {
        \foreach \y in {-0.5,0.5,1.5} {
            \fill[black] (\x,\y) circle[radius=2pt];
        }
    }
    \node at (-0.5,-0.8) {\small $-\frac{1}{2}$};
    \node at (0.5,-0.8) {\small $\frac{1}{2}$};
    \node at (1.5,-0.8) {\small $\frac{3}{2}$};
    \node at (-0.8,-0.5) {\small $-\frac{1}{2}$};
    \node at (-0.8,0.5) {\small $\frac{1}{2}$};
    \node at (-0.8,1.5) {\small $\frac{3}{2}$};
    \node at (0.5,-1.3) {\small $X_1$}; 
    \node at (-1.2,0.5) {\small $X_2$};
    \node at (0.5,-1.7) {\small Distribution~$\distq^{\theta=\frac{1}{2}}$};

    \begin{scope}[shift={(5,0)}] 
        \foreach \x in {-0.5,0.5,1.5} {
            \foreach \y in {-0.5,0.5,1.5} {
                \fill[black] (\x,\y) circle[radius=2pt];
            }
        }
        \fill[red] (-0.5,1.5) circle[radius=2pt];
        \fill[red] (0.5,0.5) circle[radius=2pt];
        \fill[blue] (1.5,0.5) circle[radius=2pt];
        \fill[blue] (0.5,-0.5) circle[radius=2pt];
        \fill[blue] (-0.5,0.5) circle[radius=2pt];
        \fill[blue] (0.5,1.5) circle[radius=2pt];
        \fill[red] (1.5,-0.5) circle[radius=2pt]; 
        \node at (-0.5,-0.8) {\small $-\frac{1}{2}$};
        \node at (0.5,-0.8) {\small $\frac{1}{2}$};
        \node at (1.5,-0.8) {\small $\frac{3}{2}$};
        \node at (-0.8,-0.5) {\small $-\frac{1}{2}$};
        \node at (-0.8,0.5) {\small $\frac{1}{2}$};
        \node at (-0.8,1.5) {\small $\frac{3}{2}$};
        \node at (0.5,-1.3) {\small $X_1$}; 
        \node at (-1.2,0.5) {\small $X_2$};
        \node[red] at (-0.5,1.7) {\small $-\alpha$};
        \node[red] at (0.5,0.7) {\small $-2\alpha$};
        \node[blue] at (1.5,0.7) {\small $+\alpha$};
        \node[blue] at (0.5,-0.3) {\small $+\alpha$};
        \node[blue] at (-0.5,0.7) {\small $+\alpha$}; 
        \node[blue] at (0.5,1.7) {\small $+\alpha$};
        \node[red] at (1.5,-0.3) {\small $-\alpha$}; 
        \node at (0.5,-1.7) {\small Distribution~$\dist^{\theta=\frac{1}{2}}$};
    \end{scope}
\end{tikzpicture}
\caption{\small{Increasing interdependence according to existing orders ($\alpha>0$).}}
\label{fig:example}
\end{figure}
Notice that when the signals become more correlated in this way, the probability that a player with a signal $\frac12$ assigns to the other player receiving a signal $\frac{-1}{2}$ or $\frac12$---and therefore the other player running---decreases. Thus, we may no longer be able to support an equilibrium where players run on $\frac{-1}{2}$ and $\frac12$. In a nutshell, in this simple example, increasing correlation of information decreases the maximal (across all equilibria) expected number of people that would run on the bank!

Why do we get this intuitively implausible result? 
The reason is that in a game of incomplete information, players' incentive constraints---and hence their choice of actions---are affected by their \textit{conditional beliefs}, and an increase in correlation (or other commonly used stochastic orders such as supermodular order or positive quadrant dependence)\footnote{There is a large literature in statistics and economics on stochastic orders. See \cite{meyer1990interdependence,meyer2012increasing,meyer2015beyond, muller2002comparison,muller2000some}} only increases the joint probability of agents receiving similar information, but does not necessarily affect agents’ conditional beliefs in the natural way that drives more agents to run. 

Indeed, this example does not invalidate our intuition that increased information homogenization should encourage agents to follow the same strategy as others in coordination games. Rather, it highlights that correlation may not be the right notion to talk about homogenization of information, and we need to think about what type of homogenization of information causes agents to follow the same strategy as others. Formally, we want a notion of similarity to compare joint distributions of signals which satisfies the following intuitive property. In any canonical binary-action coordination game of incomplete information, 
\begin{quote}
\textit{if an agent is willing to play the same strategy as all the other agents, then she should still be willing to play that strategy, if the private information across agents became more similar under this specified notion.}
\end{quote}
In other words, the set of symmetric pure strategy Bayes-Nash equilibria should expand with greater information similarity. 

The first contribution of this paper is to introduce a new class of orders---Concentration Along the Diagonal (CAD) orders---that  rank the similarity of joint distributions by comparing \textit{conditional beliefs}, which is what affects behavior in strategic situations. Roughly speaking, an increase in a CAD order means that, conditional on receiving information, each agent believes that it is now more likely that others have also received the same information. The CAD orders indeed satisfy the above notion of increasing information similarity in strategic settings. In the main results of the paper, we establish that CAD orders not only satisfy this intuitive property, but also \textit{characterize} it: We prove that increasing information similarity in a CAD order (weakly) is \textit{equivalent} to expanding the set of symmetric equilibria in a canonical class of binary-action coordination games of incomplete information. We then use the CAD orders for other games, such as congestion, collective action, and auctions.

 Formally, 
 we define an information structure to be the joint distribution of private signals (types) that agents observe. 
 We consider finite and ordered signals. We introduce two different orders to compare these joint distributions. Our first notion is called  \emph{Concentration Along the Diagonal (CAD)}. We say an information structure~$\dist$ is more similar than $\distq$, or greater in the CAD order, if any agent $i$ believes, conditional on her realized signal $s$, that it is more likely under $\dist$ than $\distq$, that any other agent $j$ also has observed exactly the same signal $s$.\footnote{It turns out that CAD is equivalent to an order proposed by \cite{meyer1990interdependence}.} Notice that an increase in similarity, as measured by the CAD order, requires that any two agents see \emph{exactly the same signal} with a higher probability. This may be too strong in some contexts, making the orders quite incomplete. Indeed, two information structures $\dist$ and $\distq$ are not comparable even if the probability of the signals being very close in value was higher under $\dist$, and the probability of signals being very different in value was lower under $\dist$. This motivates our definition of a second order, the \emph{contour-set CAD} order, which uses the order structure of the signal set. Roughly speaking, increased similarity in the contour-set CAD order means that the signals are close to each other in value with a higher probability.

Motivated by our premise---that more similar information should expand the set of symmetric pure strategy Bayes-Nash equilibria in canonical binary action coordination games---we consider a canonical class of binary-action coordination games with multiple players, in which a player's payoff depends on the aggregate actions of other players 
and on a payoff-relevant state
. Before taking their action, players receive a signal about the state drawn according to the information structure. 
The payoff difference between taking the two actions is of the form
$$\payoff(\ractoth,\omega_i)
=\a(\omega_i)+\b(\omega_i) h(\ractoth),$$
where $\ractoth$ denotes the aggregate actions of other players, $\omega_i$ is player~$i$'s signal, $\beta(\cdot)\ge 0$ and $h$ is an increasing affine function. Here, $\omega_i$ can be a private payoff-relevant signal making this setting a private-value coordination game. We also allow common-value coordination games, in which the payoff difference depends on a common payoff-relevant state~$\omega_i=\omega$ across agents, with agents receiving private signals about $\omega$. Many important economic applications share this structure.
For the private value case, one can think about $N$ firms choosing whether or not to invest in a new technology. Each firm’s payoff from investing increases with the aggregate investment and her type that captures the value of the technology to her. An example of a common-value setting is a bank run game.\footnote{See \cite{morris2016common} for examples of both private and common value cases.}  

Our first main result, Theorem~\ref{Theorem: wCAD equivalence}, shows that increasing similarity of the information structure in the CAD order is in fact equivalent to expanding the set of symmetric equilibria in this class of private-value coordination games. To understand the argument, consider the set of signals $K$ for which agents take a certain action in equilibrium. Consider an agent who receives a signal $s\in K$. If information becomes more CAD-similar, then this agent believes that another player is more likely to see a signal in the set $K$ that contains her own signal. This makes her more optimistic about the expected number of other players who will play the same action. Therefore, she has a higher incentive to play the same strategy as others, thus expanding the set of symmetric pure strategy BNE. Conversely, if information structure changes in a way that is not more CAD-similar, then there exists a set $K$ and $s\in K$ such that an agent who receives a signal $s$ believes that another agent is less likely to see a signal in set $K$ that contains her signal. We can then construct a coordination game in which there is an equilibrium where agents play a certain action when they see a signal in set $K$. However, when information structure changes, in this game, an agent with signal $s$ no longer wants to play the same action, thus eliminating this equilibrium. We also establish that CAD-similarity is equivalent to increasing (decreasing) the maximal (minimal) number of agents playing a certain action across all equilibria.

Our second result~Theorem~\ref{Theorem: CCAD with state} establishes that increasing similarity of the information structure in the contour-CAD order expands the set of symmetric monotone equilibria in the class of common-value coordination games. The converse is also true under an additional regularity condition. We also establish that under this regularity condition, in common-value coordination games, the maximal participation across all equilibria increases and the minimal participation decreases if and only if information becomes more contour-CAD similar.

These two main results speak about how information similarity makes it \textit{possible} for agents to play the same strategy in equilibrium. One may be interested in whether more CAD similarity improves welfare. In general, coordination games have multiple equilibria, and so the effect on welfare depends on which equilibrium is selected. For instance, in the bank run example,  suppose that the social planner wants to minimize the probability of a run. If the planner anticipates the agents will play the adversarial (maximal run) equilibrium, then she prefers lesser homogenization of information.  Conversely, if the planner anticipates the agents will play the advantageous (minimal run) equilibrium, then she prefers greater homogenization of information.

The two characterization results, albeit for a specific class of games, serve as a powerful proof-of-concept for the CAD orders. An interesting line of inquiry is to study the effect of increasing similarity of information in the CAD order in other strategic settings beyond affine coordination games. We consider coordination settings in which an agent's incentive to take an action increases with the aggregate action but not necessarily in an affine fashion, i.e., $h(\cdot)$ is an increasing function but not affine. We develop a stronger notion of CAD, called sCAD, and establish a similar equivalence result. We then consider games of strategic substitutability like congestion games, and show that, analogous to our main results, more CAD-similar information shrinks the set of symmetric equilibria in congestion games. In an auction setting, we show that the revenue from a second-price auction increases when players’ valuations become more similar in the CAD order. In \cite{basak2024protest}, we consider collective action games---games with both strategic complementarities and substitutabilities---and show that CAD yields a neat characterization of when increased similarity of information facilitates or hinders collective action.

\subsection{Related Literature}
Our paper relates to a literature in economics and statistics on comparing interdependence of random variables, \cite[see][for instance]{muller2002comparison,meyer2015beyond,meyer2012increasing,meyer1990interdependence}. \cite{meyer2012increasing} and \cite{meyer2015beyond} illustrate the usefulness of various existing dependence orders, particularly the supermodular order, for economic applications. \cite{epstein1980increasing} provide a behavioral foundation based on preferences for correlation for the supermodular order in the case of bivariate random variables.\footnote{Strictly speaking, they focus on the lower orthant order, which is  equivalent to the supermodular order in two dimensions.} 
However, most  of the existing orders do not 
compare the \emph{conditional} belief distributions that arise in strategic settings with incomplete information.\footnote{It is worth mentioning that the CAD orders satisfy all the desirable properties of stochastic orders of interdependence proposed by \cite{joe1997multivariate}.}

We contribute more broadly to research on games of incomplete information. A large literature that dates back at least to \cite{H71}  studies how exogenous changes in the information environment impacts behavior in games. More recently, \cite{morris2002social}, \cite{angeletos2007efficient}, \cite{bergemann2013robust}, \cite{jensen2018distributional}, and \cite{mekonnen2022bayesian} have studied this in a class of games similar to ours where information is dispersed among agents. 
These studies focus on the value of new information, whereas our setup has no new information. Instead, the information becomes more or less similar. When there is a common state, as is typical in the above-mentioned papers, similarity of information affects the agents' strategic uncertainty but not their fundamental uncertainty (see Remark \ref{remark:commonbelief}). 
Unlike in the global games literature, our focus is not on equilibrium uniqueness, but rather we show how the set of equilibrium changes with similarity of information. \cite{gossner2000comparison}, \cite{cherry2012strategically}, and  \cite{bergemann2016bayes} have also proposed ranking of information structures based on equilibrium set inclusion. 

We consider the effect of changing interdependence of multivariate random variables while keeping fixed the marginal distributions. \cite{clemen1985limits} and \cite{CB24} study how such changes impact the value of information and show that informational diversity may be valuable. \cite{de2023robust} consider an environment with known marginal distributions but unknown joint distribution to obtain the robustly optimal policy in a class of decision problems. \cite{cripps2008common} and \cite{awaya2022common} study the effect of the interdependence of signals on common learning. 
\cite{awaya2022common} show that when agents receive information over multiple periods, then 
more interdependence of information within period 
obstructs common learning. This happens because increasing interdependence within a period does not imply that the signals across multiple periods become more
similar.\footnote{See Online Appendix Section~\ref{Section: awaya krishna example} for a more detailed discussion.}

\section{Notions of Information Similarity}\label{sec:CAD orders}
Since we are interested in studying similarity in settings with incomplete information and multiple agents we define an information structure to be a joint distribution of agents’
private signals or types.\footnote{Henceforth, we will use the terms~``types" and ``signals" interchangeably.} We define two orders of similarity of information structures.  The notions stem from the basic idea that more similar information means that conditional on observing a private signal~$s$, a player believes any other player is more likely to have observed exactly the same signal~$s$. This motivates the name for our class of orders ``Concentration along a Diagonal." 

 Let $\sigset$ be a fixed, finite set of types.  Let $\vec \sigr=(\sigr_1,\sigr_2,\ldots ,\sigr_\nplayers)$ denote the $\sigset^\nplayers-$valued random variable that represents the types of $N$ agents. We assume that $\vec \sigr$ is distributed according to some exchangeable distribution~$\dist$. 
 For any $i$ and $j\neq i$, and any $K \subset \sigset$, we define $\dist_\signal(\cdot) \in \De(\sigset)$ by $$\dist_\signal(K):=Prob(\sigr_j\in K \vert \sigr_i = \signal).$$ In particular, $\dist_\signal(s')$ means $\dist_\signal(\{\sigr_{j} = s'\vert \sigr_i = \signal\})$. Exchangeability of $\dist$ implies that we need not index these conditional distributions with player identities.

\subsection{Concentration along a Diagonal}\label{sec:wCAD}
\bdefn[\textbf{Concentration along a Diagonal}]\label{Definition: Weak CAD} We say that $\dist$ has a ``higher concentration along a diagonal'' than $\distq$, or $\dist$ \textbf{is CAD higher than} $\distq$, denoted by $\dist \cad \distq$, if, 
\begin{enumerate}
    \item $\dist$ and $\distq$ have the same marginal distributions. And, 
    \item For all $\signal \in \sigset$, and any~$i,j$ \begin{enumerate}[(a)]
        \item $\cdist(\signal) \ge \cdistq(\signal)$, and
        \item $\cdist(\signal') \le \cdistq(\signal')$ whenever $\signal' \neq \signal$.
    \end{enumerate}
\end{enumerate}
\edefn 
\noindent
If $\sigr \sim \dist$ and $\sigry \sim \distq$ with $\dist \cad \distq$, we say $\sigr \cad \sigry$, i.e., we use $\sigr \cad \sigry$ and $\dist\cad \distq$ interchangeably. 

In words, when information becomes more similar in the CAD order, any agent~$i$ believes, conditional on being realized type~$s$, that it is more likely that any other agent~$j$ is also of exactly the same type~$s$. It is worth noting that verifying whether two distributions are ranked in this order is not computationally hard. For exchangeable distributions, the complexity is $\mathcal O(\vert \sigset\vert^2)$. The CAD order is equivalent to the orders in \cite{meyer1990interdependence} \cite[see Proposition~1 and ~5 from][]{meyer1990interdependence}.

To gain some more intuition about the CAD order, we provide some equivalent formulations  in the lemma below. It is almost immediate that an increase in the CAD order is equivalent to requiring that, conditional on being of type~$s$, any agent~$i$ assigns higher probability to another agent~$j$ having type in \emph{any subset} of types  that includes type~$s$.
Another way to formulate the order is in terms of the number of other agents who have the same type as a given agent~$i$. An increase in similarity in the CAD order means that, conditional on any agent~$i$ being of type~$s$, the expected number of other agents with type in any subset of types that includes~$s$ is higher. We state this below. The proof is in the appendix. Define, for any agent~$i$, for any $K \subseteq \sigset$, 
\begin{align*}
    \countr(K) := \sum_{j\neq i} \ind_{\sigr_j \in K}
\end{align*}
$\countr(K)$ counts the number of players other than player $i$ with realized type in set $K$. $\countr(K)$ is a $\{0,1,\ldots,\nplayers-1\}$-valued random variable. Let $\ccdfi^{\signal,K}$ denote the CDF of $\countr(K)$ conditional on $\sigr_i = \signal$, and $\sigr$ is distributed according to $\dist$. 

\blemma\label{Lemma: wCAD equivalent with increasing expectation} Let $\sigr$ and $\sigry$ be two $\sigset^\nplayers$-valued, exchangeable random variables with distributions $\dist$ and $\distq$ respectively. Moreover, $\dist$ and $\distq$ have identical marginals. Then, the following are equivalent. 
\begin{enumerate}
    \item $\dist \cad \distq$.
    \item $Prob(\sigr_j \in K \vert \sigr_i = \signal) \ge Prob(\sigry_j \in K \vert \sigry_i = \signal) \quad \forall i,j, \signal, K \ni \signal.$
     \item For all $\signal \in \sigset$ and $K \subseteq \sigset$ such that $\signal \in K$, and for all $i$,  $$\E\left[ \countr(K) \bigg\vert \sigr_i = \signal \right] \ge \E\left[\countr(K)\bigg\vert \sigry_i = \signal\right].$$
     
\end{enumerate}   
\elemma 
\noindent

An increase in similarity, as measured by  CAD, requires that any two agents see \emph{exactly the same} signal with a higher probability. This may be too strong in some contexts, making the orders quite incomplete. Suppose that the signals are ordered. Indeed, two information structures $\dist$ and $\distq$ are not comparable even if the probability of events in which $\dist$ and $\distq$ are \emph{very close} in value is higher under $\dist$, and the probability of events in which $\dist$ and $\distq$ are \emph{very different} in value is lower under $\dist$. 

\subsection{Concentration along a Diagonal: Contour Sets}\label{sec:cCAD}
We introduce a weaker (less incomplete) order of similarity called \emph{Contour-set CAD} that uses the ordering of signals. Intuitively, increased similarity according to the contour-set CAD order, no longer means that there is a higher chance of getting exactly
the same signal. Rather, $\dist$ is more similar than $\distq$ if the signals are close to each other in value with a higher probability under $\dist$ than under $\distq$. 

Formally, assume that the set of signals~$\sigset$ is an ordered set with $\signal_1 \le \signal_2 \le\ldots \le \signal_n$. Given any $\sighat \in \sigset$, we define upper and lower contour sets of $\sighat$. 
$$\sighup := \{ y \in \sigset: y \ge \sighat\}.$$
$$\sighdn := \{ y \in \sigset: y \le \sighat\}.$$
Recall that the notation $\dist_\signal(s')$ means $\dist_\signal(\{\sigr_{j} = s'\vert \sigr_i = \signal\})$.
\bdefn[\textbf{Concentration along a Diagonal: Contour Sets}]\label{Definition: CCAD} We say that $\dist$ has a ``higher concentration on the contour sets along a diagonal'' than $\distq$, or $\dist$ \textbf{is cCAD higher than} $\distq$, denoted by $\dist \ccad \distq$, if, 
\begin{enumerate}
    \item $\dist$ and $\distq$ have the same marginal distributions. And, 
    \item For all $\signal \in \sigset$, \begin{enumerate}[(a)]
        \item $\cdist(\sighup) \ge \distq(\sighup)$ for all $\sighat \le \signal$, and
        \item $\cdist(\sighdn) \ge \distq(\sighdn)$ for all $\sighat \ge \signal$
    \end{enumerate}
\end{enumerate}
\edefn
The contour set CAD order can be stated in terms of intervals. It turns out that $\dist$ has a higher concentration on the contour sets along a diagonal than $\distq$, if and only if conditional on being of type~$s$, an agent~$i$ assigns higher probability to any other agent~$j$ having type in \emph{any interval} of types that includes type~$s$. 
\bdefn[\textbf{Concentration along a Diagonal: Intervals}]\label{Definition: ICAD} We say that $\dist$ has a ``higher concentration on the intervals along a diagonal'' than $\distq$, or $\dist$ \textbf{is iCAD higher than} $\distq$, denoted by $\dist \icad \distq$, if, 
\begin{enumerate}
    \item $\dist$ and $\distq$ have the same marginal distributions. And, 
    \item For all $\signal \in \sigset$, $\cdist(K) \ge \cdistq(K)$ for all the intervals $K \ni \signal$. 
\end{enumerate}
\edefn

\noindent We state the equivalence of $\ccad$ and $\icad$ formally in Proposition \ref{Proposition: relation between orders}. 

\subsection{Comparing the notions}\label{sec:compare}
The next result characterizes the relationship between these orders. The relationship is as a reader might expect. 
\bprop\label{Proposition: relation between orders}  Let $\sigr,\sigry$ be two $\sigset^\nplayers$ valued random variables, where $\sigset$ is an ordered set.
$$  \sigr \cad \sigry \implies \sigr \ccad \sigry \Longleftrightarrow \sigr \icad \sigry.$$
\eprop

\bprf[Proof of Proposition~\ref{Proposition: relation between orders}]
The result that $\sigr \cad \sigry \implies \sigr \ccad \sigry$ follows by invoking Lemma~\ref{Lemma: wCAD equivalent with increasing expectation} and letting $K$ be any upper- or lower-contour sets. Since contour sets are intervals
$\sigr \icad \sigry \implies \sigr \ccad \sigry$ is also trivially true.

Next, we prove $\sigr \ccad \sigry \implies \sigr \icad \sigry$. Suppose that $\sigr \sim \dist$ and $\sigry \sim \distq$ with $\dist \ccad \distq$,  but $\dist \nicad \distq$. Therefore, $\exists$ $\signal \in \sigset$ and an interval $K \ni \signal$ such that $\cdist(K) < \cdistq(K)$. Let $\slb, \sub$ be the min and max elements of $K$ respectively. We have, 
\begin{align*}
    1 =& \cdist(\{\sigy: \sigy < \slb\}) + \cdist(K) + \cdist(\{\sigy: \sigy > \sub\}) \\
    =& \cdistq(\{\sigy: \sigy < \slb\}) + \cdistq(K) + \cdistq(\{\sigy: \sigy > \sub\})
\end{align*}
Therefore, at least one of the following holds: 
\begin{enumerate}[$(i)$]
    \item $\cdist(\{\sigy: \sigy < \slb\}) > \cdistq(\{\sigy: \sigy < \slb\})$.
    \item $\cdist(\{\sigy: \sigy > \sub\}) > \cdistq(\{\sigy: \sigy > \sub\})$
\end{enumerate}
If $(i)$ holds, then $\cdist(\slb^\uparrow) < \cdistq(\slb^\uparrow)$. If $(ii)$ holds, then $\cdist(\sub^\downarrow) < \cdistq(\sub^\downarrow)$. In either case, this contradicts $\dist \ccad \distq$.

To see that $\ccad$ does not imply $\cad$, consider Figure~\ref{fig:example_ccad_not_wcad}. We start with a random variable $\sigr$ with a support $\{1,2,3,4\}\times \{1,2,3,4\}$. We construct a random variable $\sigry$ by increasing the mass on realizations $(1,1), (2,3), (3,2),$ and $(4,4)$ by $\a > 0$ each, and reducing the mass on realizations $(1,3), (2,4), (3,1),$ and$(4,2)$ by $\a$. The marginal distributions of $\sigr$ and $\sigry$ coincide. Moreover, $\sigr \ccad \sigry$ but $\sigr\ncad \sigry$. 
\begin{figure}[h!]
    \centering
    \begin{tikzpicture}[scale=1.2]
        \foreach \x in {1,2,3,4}
        \foreach \y in {1,2,3,4}
        {
        \fill[black] (\x,\y) circle[radius=2pt];
        }
        \node at (2.5,-0.25) {\small $\sigr_1$};
        \node at (-0.25,2.5) {\small $\sigr_2$};
        
        \foreach \x in {1,2,3,4}
        {
        \node at (\x,0.5) {\small \x};
        }

        \foreach \y in {1,2,3,4}
        {
        \node at (0.5,\y) {\small \y};
        }
        
        \fill[red] (1,3) circle[radius=2pt];
        \node at (1.3,3) {\small $-\alpha$};

        \fill[red] (2,4) circle[radius=2pt];
        \node at (2.3,4) {\small $-\alpha$};

        \fill[red] (3,1) circle[radius=2pt];
        \node at (3.3,1) {\small $-\alpha$};

        \fill[red] (4,2) circle[radius=2pt];
        \node at (4.3,2) {\small $-\alpha$};
        
        \fill[blue] (1,1) circle[radius=2pt];
        \node at (1.3,1) {\small $+\alpha$};

        \fill[blue] (2,3) circle[radius=2pt];
        \node at (2.3,3) {\small $+\alpha$};

        \fill[blue] (3,2) circle[radius=2pt];
        \node at (3.3,2) {\small $+\alpha$};

        \fill[blue] (4,4) circle[radius=2pt];
        \node at (4.3,4) {\small $+\alpha$};

    \end{tikzpicture}
    \caption{$\sigry \ccad \sigr$ but $\sigry \protect\ncad \sigr$.}
    \label{fig:example_ccad_not_wcad}
\end{figure}
\eprf 
In the appendix \ref{subsection: relation to other orders}, we discuss the relationship between the CAD orders and other well-known orders such as supermodular order or concordance order. For $\nplayers  = 2$, our orders imply the supermodular and positive quadrant dependence order, and for $\nplayers  > 2$ the notions are not nested.

\section{CAD Similarity and Coordination Equivalence}\label{sec: main results}

The example in the introduction highlighted, somewhat surprisingly, that an increase in correlation (or in the supermodular order or positive quadrant dependence) does not increase the incentives for an agent to play the same strategy as others in simple coordination games. 
In this section, we show that in a canonical class of binary-action coordination games, increasing information similarity across agents in the sense of increasing CAD is \emph{equivalent} to expanding the set of equilibria.\footnote{\cite{bergemann2016bayes} also use equilibrium set inclusion to define a partial order of informativeness in games of incomplete information. The authors propose that more information shrinks the set of Bayes Correlated Equilibria because if agents have more information, then a smaller set of outcomes is incentive-compatible.}

The canonical class of binary-action coordination games we consider has been widely used in macroeconomics \cite[e.g.,][for a currency attack]{morris1998unique}, financial economics \cite[e.g.,][for bank-runs and financial fragility]{goldstein2005demand, vives2014strategic}, industrial organization \cite[e.g.,][for a platform network externality]{farrell2007coordination}, political economy \cite[e.g.,][for protests]{shadmehr2011collective}, development economics \cite[e.g.,][for under-investment problem]{sakovics2012matters}, international economy \cite[e.g., ][for tax havens]{konrad2016coordination}, and organizational economics \cite[e.g.,][for reporting sexual misconduct problem]{cheng2022reporting}. We consider two kinds of coordination games: private value and common value. The main difference is that the first kind does not have a common state whereas the second kind does. \cite{morris2016common} first highlighted this distinction.

\subsection{Binary-action private-value coordination games}\label{subsec:Game definition}
Consider a setting in which  $\nplayers$ players, indexed by $i$ or $j\neq i$, simultaneously and independently choose whether to take an action ($a_i=1$) or not ($a_i=0$). Each player $i$'s payoff depends on the aggregate action by others $$\ractoth=\sum_{j\neq i} a_j.$$
Using the same notation as in Section~\ref{sec:CAD orders}, we let $\sigset$ be a fixed, finite set from which player types are drawn. Let $\vec \sigrs=(\sigrs_1,\sigrs_2,\ldots ,\sigrs_\nplayers)$ be a $\sigset^\nplayers-$valued random variable distributed according to an exchangeable distribution~$\dist$. Suppose that $\vec s=(s_1,s_2,\ldots s_\nplayers)\in \sigset^\nplayers$ is the realized type profile. Player~$i$ knows her own type $\signal_i$ but does not know a different player~$j$'s type $\signal_j$. As before, for any two agents~$i$ and $j\neq i$, and any $K \subset \sigset$ define $\dist_\signal(\cdot) \in \De(\sigset)$ by 
$$\dist_\signal(K):=Prob(\sigrs_j\in K \vert \sigrs_i = \signal).$$ In particular, $\dist_\signal(s')$ means $\dist_\signal(\{\sigrs_{j} = s'\})$. 

A player $i$ with type $\sigrs_i=s_i$ gets a payoff of $u(a_i=1,\ractoth,\signal_i)$ if she acts and $u(a_i=0,\ractoth,\signal_i)$ if she does not act. The net payoff from taking action~$a=1$ is
\begin{align}
    \payoff(\ractoth,\signal_i):=u(a_i=1,\ractoth,\signal_i)-u(a_i=0,\ractoth,\signal_i)=\a(\signal_i)+\b(\signal_i)h(\ractoth)\label{Equation: payoff difference games}
\end{align}
where $h(\cdot)$ is increasing and $\b(\cdot)$ is non-negative. Since $\beta(\cdot)\ge 0$ for all $s\in \sigset$, the game~$\G$ exhibits strategic complementarity. Henceforth, we refer to such symmetric binary-action games of strategic complementarity simply as \textit{private-value coordination games}. 

\begin{definition}[\textbf{Affine Private-value Coordination games}]
A game~$\G$ a called an \textbf{affine private-value coordination game} when the payoff difference can be described by Equation \eqref{Equation: payoff difference games} and $h(.)$ is affine. 
\end{definition}
\bass\label{Assumption: signals are payoff relevant}\textbf{Payoff-relevant signals:} For every $i \in \players$ and $\signal_i,\signal'_i\in \sigset$ such that $\signal_i\neq \signal'_i$,  there is $\actoth$ such that $\payoff(\actoth, \signal_i) \neq \payoff(\actoth,\signal'_i)$.\eass 
This assumption simply precludes the possibility of artificially creating  finer signal spaces from the original signal space 
to render the notion of changing similarity vacuous.\footnote{For instance, in bank run example, a new signal space can be created by ``splitting'' each signal into two signals by adding a payoff irrelevant feature, such as color. Then, a CAD increase in the original signal space may not translate to a CAD increase in the new signal space.
}

Technology adoption by firms is an example of such a game. To fix ideas, consider $N$ firms choosing whether or not to invest in a new technology. A firm's type~$s_i$ captures the value of the new technology to her. Each firm's payoff from investment depends also on the aggregate investment. The non-negative $\beta(\cdot)$ captures the network externality \cite[see][]{katz1985network, farrell1986installed}, where each firm benefits when more firms invest in the new technology.

\subsection*{\small{Equilibrium}} 
We restrict attention to symmetric Bayes Nash Equilibrium in pure strategies (henceforth, equilibrium). Under distribution $\dist$, a strategy profile where each player plays $\strategy:\sigset\to\{0,1\}$ constitutes an equilibrium if
\begin{align}
    \strategy(\signal) = 1 \implies \E[\payoff(\ractoth, \signal) \vert \signal,\strategy;\dist ] \ge 0 \tag{IC:P}\label{Equation: IC to play a=1}\\
    \strategy(\signal) = 0 \implies \E[ \payoff(\ractoth, \signal) \vert \signal,\strategy;\dist ] \le 0 \tag{IC:NP}\label{Equation: IC to play a=0}
\end{align}
Let $\mathcal{E}(\G, \dist)$ be the set of equilibrium in the game $\G$ under distribution $\dist$. The restriction to symmetric pure strategies implies that any strategy~$\sigma$ simply partitions the signal space into participation and non-participation sets---sets of signals after which agents take action  and do not take action respectively. 
\begin{equation*}
    P(\strategy): = \{\signal\in\sigset: \strategy(\signal) = 1\}, \ 
    NP(\strategy) := \sigset \backslash P(\strategy).
\end{equation*} 
Define the maximal and minimal participation across all equilibria as follows
$$\maxP(\G,\dist) := \max_{\strategy \in \eq(\G,\dist)} \dist(P(\strategy)) \quad \text{ and } \quad \minP(\G,\dist):= \min_{\strategy \in \eq(\G,\dist)} \dist(P(\strategy)).$$

\subsection{CAD Similarity and Coordination Equivalence}
Our first result, Theorem~\ref{Theorem: wCAD equivalence}  shows that increasing similarity of information in the CAD order is equivalent to 
aiding agents playing the same strategy in the class of coordination games~$\G$ when $h(\cdot)$ is affine. 
\bthm\label{Theorem: wCAD equivalence} Let $\dist$ and $\distq$ be two joint distributions over $\sigset^N$. The following are equivalent. 

\begin{enumerate}
    \item $\dist \cad \distq$.  
    \item $\eq(\G, \dist) \supseteq \eq(\G,\distq)$ for all affine private-value coordination games~$\G$.
    \item $\maxP(\G, \dist) \ge \maxP(\G,\distq)$ for all affine private-value coordination games~$\G$.
    \item $\minP(\G, \dist) \le \minP(\G,\distq)$ for all affine private-value coordination games~$\G$.
\end{enumerate}
\ethm

The formal proof is in the appendix. Below, we discuss the main argument. We first show that an equilibrium $\sigma$ remains an equilibrium when information become more CAD-similar (from $\distq$ to $\dist$). Given equilibrium $\sigma$, let $P(\strategy)$ and $NP(\strategy)$ be the set of signals under which an agent participate and do not participate, respectively. Consider any agent $i$ with signal $s$. When information becomes more CAD similar, her incentive to participate changes as follows
\begin{align*}
\E[ \payoff(\ractoth, \signal) \vert \signal,\strategy;\dist ] - \E[ \payoff(\ractoth, \signal) \vert \signal,\strategy;\distq ] \\
=\b(s) \left[ \E[ h(\ractoth) \vert \signal,\strategy;\dist ] - \E[h(\ractoth) \vert \signal,\strategy;\distq ] \right].
\end{align*}
Recall that $\beta(s)\geq 0$, $h(\cdot)$ is affine and increasing, and $\ractoth$ is the number of agents who receives signals in $P(\sigma)$. By Lemma~\ref{Lemma: wCAD equivalent with increasing expectation}, an agent with $\signal \in P(\strategy)$ expects more agents to see a signal in $P(\strategy)$ that contains her own signal. Therefore, she has a higher incentive to participate than before (the above expression is non-negative). Conversely, an agent with $\signal \in NP(\strategy)$ expects fewer agents to see a signal in $P(\strategy)$ that does not contain her own signal. Therefore, she has a lower incentive to participate than before (the above expression is non-positive). Therefore, if $\sigma$ is an equilibrium under $\distq$ then it remains an equilibrium under more CAD similar $\dist$. This establishes $(1) \implies (2)$. If the set of equilibrium expands, then the maximal participation increases and the minimal participation decreases, which establishes  $(2) \implies (3), (4)$. 

Finally, we establish that $(3) \implies (1)$. A similar argument holds for $(4)\implies (1)$. Suppose that (3) holds, but $\dist \ncad \distq$. Then, $\exists \signal^*\in \sigset$ and $K \ni \signal^*$ such that, 
$$\dist(\sigrs_j \in K \vert \sigrs_i = \signal^*) 
< \distq(\sigrs_j \in K \vert \sigrs_i = \signal^*).$$
Consider a strategy profile~$\strategy$ with participation set $P(\strategy)=K$. 
We can construct a game $\G$ in which  $\strategy$ is the maximal equilibrium under information structure $\distq$. To do this, we choose $\alpha(s)$ sufficiently negative for $s\notin K$, so that participation is a dominated strategy for $s\notin K$. Also, by construction, we make an agent with signal $\signal^*$ indifferent between participating and not participating under $\distq$. Now, consider changing the information structure to $\dist$. It follows from the above inequality that the agent with signal $\signal^*$ now strictly prefers not participating. This means $\sigma$ is no longer an equilibrium under $\dist$. Recall that it is a dominated strategy for $s\notin K$ to participate. Therefore, under $\dist$, there is no  other equilibrium with greater participation. Since given $\strategy$, aggregate participation depends only on the marginal distribution, we must have 
$\maxP(\G, \dist)<\maxP (\G,\distq)$, which contradicts (3).

\subsection{Common-value Affine Coordination Games}
Theorems~\ref{Theorem: wCAD equivalence} establishes an equivalence between increases in similarity in CAD orders and the expansion of the set of \emph{all} pure strategy equilibria. But in many economic applications, it is reasonable to consider some structure in the equilibrium strategies: For instance, 
in this section, we study a class of coordination games commonly studied in the global games literature, in which it is typical to focus on equilibria in cutoff strategies--- a player acts if and only if her type is high enough.\footnote{In fact, under certain regularity conditions, the maximal and minimal equilibria in these settings are in cutoff strategies   \cite[see][for a survey]{Morris_Shin_2003}.} We establish an analogous equivalence between increasing similarity in the contour-CAD order and expanding the set of cutoff equilibria.

As in Section~\ref{subsec:Game definition}, we consider a setting in which $N$ players simultaneously and independently choose whether to act $(a_i=1)$ or not $(a_i=0)$. But now, each player~$i$'s payoff depends on the aggregate action by others and an unknown but common underlying state~$\stater\in \states$, where $\states$ is a finite ordered set. The state is drawn from a common prior $\posterior_0\in \Delta (\states)$. We interpret an agent’s type as the private signal she receives about the underlying state. 
Let $\sigset$ be a fixed, finite ordered set of signals, and the profile of signals $\vec \sigr$ be a $\sigset^\nplayers-$valued random variable. Conditional on $\stater=\state$, $\vec \sigr$ is distributed according to $\dist^\state$. We assume that $\dist^\state$ is exchangeable for all $\state \in \states$. Therefore, $\dist$ is also exchangeable. 

A player $i$ receives a signal $\sigr_i$. The marginal distribution of $\sigr_i$ conditional on state $\stater=\state$ is denoted by $\marg \dist^\state (.)=\sum_{\signal_{-i}\in \sigset^{\nplayers-1}} \dist^\state(.,\signal_{-i})$.  Every signal realization $\sigr_i=\signal$ generates a posterior distribution over $\states$  obtained using Bayes rule. That is, for any $T \subseteq \states$, 
\begin{align*}
    \posterior(\signal)(T)  =  \frac{\sum_{\state \in T} \posterior_0(\state) \marg \dist^\state(\signal)}{\sum_{\state \in \states} \posterior_0(\state) \marg \dist^\state(\signal)}
\end{align*}
We assume that $\sum_{\state \in \states} \posterior_0(\state) \marg \dist^\state(\signal) > 0$ for all $\signal\in \sigset$.  
Conditional on $\stater=\state$, as before, for any $i$ and $j\neq i$, and any $K \subset \sigset$, we define $\dist^\state_\signal(\cdot) \in \De(\sigset)$ by $$\dist^\state_\signal(K):=Prob(\sigr_j\in K \vert \sigr_i = \signal, \stater=\state).$$ In particular, $\dist^\state_\signal(s')$ means $\dist^\state_\signal(\{\sigr_{j} = s'\})$. Exchangeability of $\dist^\state$ implies that we need not index these conditional distributions with player identities.

\noindent
The net payoff of each player takes the following form.
\begin{align}
\label{Equation: payoff difference common value games}
    \payoff(\actoth,\state) = \a(\state) + \b(\state) h(\actoth),
\end{align}
for some affine, increasing $h(\cdot)$ and $\b (\cdot) \ge 0$. We make an assumption similar to Assumption~\ref{Assumption: signals are payoff relevant}. 
\noindent
\bass\label{Assumption: payoff relevance with common state}[\textbf{Payoff-relevant state}] For every $i\in\players$ and $\state \neq\state'$, there is $\actoth$ such that $\payoff(\actoth,\state) \neq \payoff(\actoth,\state')$. 
\eass

Notice that unlike in the games described in Section~\ref{subsec:Game definition},  players' payoffs do not depend on their idiosyncratic type. Rather, this is a ``common-value" environment, where payoffs depend on a common unknown state. 

\begin{definition}[\textbf{Common-value Affine Coordination games}]
A game~$\widetilde \G$ a called a \textbf{common-value affine  coordination game} when the payoff difference can be described by Equation \eqref{Equation: payoff difference common value games}, with  increasing and affine $h(\cdot)$ and non-negative $\b(\cdot)$. 
\end{definition}

There are many applications with this payoff structure. A famous example is a currency attack game \`a la \citet{morris1998unique}. Speculators simultaneously decide whether or not to attack a fixed exchange rate regime by selling the currency short. The payoff from not shorting is $u(a_i=1,\actoth,\state)=0$ and the payoff from shorting is $u(a_i=0,\actoth,\state)=b (\theta) (1-p(\state,\actoth)) - cp(\state,\actoth),$
where $p(\state,\actoth)$ is the probability that the currency will not be devalued, $b(\theta)>0$ is the benefit when the currency is devalued, and $c$ is the cost when it is not devalued. This gives $\payoff(\actoth, \state)=(b(\theta)+c)p(\state,\actoth)-b(\state)$. Let $p(\state,\actoth)=p_\theta \state + p_A h(\actoth)$. The currency is less likely to be devalued when the currency is stronger ($p_\state>0$) and fewer speculators short ($h(\actoth)$ affine and increasing and $p_A>0$).  Then, 
$$\payoff(\actoth, \state)=\underbrace{p_\theta(b(\theta)+c) \state -b(\theta)}_{\a(\state)}+\underbrace{p_A(b(\state)+c)}_{\b(\state)>0} \cdot h(\actoth).$$
Similar payoff structures arise in other games with strategic complementarity like debt rollover problems (e.g., \cite{morris2004coordination}) or bank runs (e.g.,\cite{goldstein2005demand}) when the probability of the borrower or the bank surviving is $p(\state,\actoth)=p_\theta \state + p_A h(\actoth)$.

\subsection*{\small{Equilibrium}} 
We restrict attention to symmetric Bayes Nash Equilibrium in pure strategies (henceforth, equilibrium). Under distribution $\dist$, a strategy profile where each player plays $\strategy:\sigset\to\{0,1\}$ constitutes an equilibrium if
\begin{align*}
    \strategy(\signal) = 1 \implies \E[\payoff(\ractoth,,\stater) \vert \signal,\strategy;\dist ] \ge 0 \\
    \strategy(\signal) = 0 \implies \E[ \payoff(\ractoth, \stater) \vert \signal,\strategy;\dist ] \le 0.
\end{align*}

\subsection*{\small{Cutoff Strategies}}
As is customary in many applications, we focus on equilibria in cut-off strategies. 
\begin{definition}
We say a strategy~$\strategy$ is a cutoff equilibrium if there is a cutoff~$\tilde s$ such that $P(\strategy)=\{s\in\sigset: s\geq \tilde s\}$ and $NP(\strategy)=\sigset\setminus P(\strategy)$. Let $\eqc(\widetilde \Gamma,\dist)$ denote the set of cutoff equilibria given any common-value affine coordination game $\widetilde \G$ and information structure $\dist$. 
\end{definition}
Recall that, $\dist(\cdot) = \sum_{\state} \posterior_0(\state) \dist^\state(\cdot)$. Define \begin{align*}
    \eqmaxp(\G,\dist):= \max_{\strategy \in \eqc(\G,\dist)}\dist(P(\strategy))
\end{align*} to be the maximal equilibrium participation. Similarly, define
\begin{align*}
    \eqminp(\G,\dist):= \min_{\strategy \in \eqc(\G,\dist)} \dist(P(\strategy)) = 1- \max_{\strategy \in \eqc(\G,\dist)} \dist(NP(\strategy))
\end{align*} to be the minimal equilibrium participation. If $\eqc(\G,\dist) = \emptyset$, then we set $\eqmaxp(\G,\dist) = 0$ and $\eqminp(\G,\dist) = 1$.

\subsection{Contour-CAD and Cutoff Equilibria}\label{sec:contourcad}

We show below that in this class of games, increasing similarity in the contour CAD order implies an increase in the set of cut-off equilibria, and the converse is also true under an additional regularity condition.
\bdefn[\textbf{Affine independence}]\label{Definition: affine independent info} An information structure $\dist$ is \textbf{affinely independent} if the set of posteriors~$\mu(\signal)$ it generates is affinely independent. That is, 
\begin{align*}
    \sum_{\signal \in \sigset} \lm_\signal \posterior(\signal) = 0 \text { and } \sum_{\signal\in\sigset} \lm_\signal =0 \implies \lm_\signal = 0 \forall \signal \in \sigset.
\end{align*}
\edefn

\bthm\label{Theorem: CCAD with state} Let $\dist=(\dist^\state)_{\state\in \states}$ and $\distq=(\distq^\state)_{\state\in \states}$ be two distributions with identical marginal distributions for any $\state\in \states$. 

\begin{enumerate}
    \item If $\dist^\state \ccad \distq^\state$  for all 
    $\state$ with $\dist^\state \neq \distq^\state$ then 
    $\eqc(\widetilde \Gamma,\dist) \supseteq \eqc(\widetilde \Gamma,\distq)$ for all common-value affine coordination games $\widetilde \Gamma$. 
    \item Suppose $\distq$ (and hence $\dist$) is affinely independent. Let $T \subseteq \states$ be such that $\dist^\state \neq \distq^\state$ iff $\state \in T$. Then, the following are equivalent. 
    \begin{enumerate}
        \item $\dist^\state \ccad \distq^\state$ $\forall \state \in T$.
        \item $\eqc(\widetilde \Gamma,\dist)\supseteq \eqc(\widetilde \Gamma,\distq)$ for all common-value affine coordination games $\widetilde \Gamma$.
        \item $\eqmaxp(\widetilde \Gamma,\dist)\ge \eqmaxp(\widetilde \Gamma,\distq)$ and $\eqminp(\widetilde \Gamma,\dist)\le \eqminp(\widetilde \Gamma,\distq)$ for all common-value affine coordination games $\widetilde \Gamma$.
    \end{enumerate}    
    
\end{enumerate}
\ethm

\medskip
The proof of Theorem~\ref{Theorem: CCAD with state} is in the appendix. Below we explain why we impose the condition of affine independence of the posteriors $\posterior(\cdot)$, that we did not need in proving Theorems~\ref{Theorem: wCAD equivalence}. For part $(2)$, we want to show that if $\dist^\state \nccad \distq^\state$ for some $\state \in \states$, then we can construct a common value coordination game $\widetilde\Gamma$ such that $\eqc(\widetilde\Gamma;\distq) \nsubseteq \eqc(\widetilde\Gamma;\dist)$. If $\dist^\state \nccad \distq^\state$, then we have a signal $\signal$ and a contour set $\sighup$ such that $\signal \in \sighup$ and $\cdist^\state(\sighup) < \distq^\state(\sighup)$.\footnote{
Alternatively, a contour set $\sighdn$ with an analogous implication that we choose to omit in this discussion for the sake of brevity.} To prove the analogous parts of Theorems~\ref{Theorem: wCAD equivalence}, we used the dependence of the payoff on the private signal $\beta(\cdot)$ to construct a game such that the desired equilibrium set inclusion fails. Now, $\beta(\cdot)$ depends only on the common $\state$, and so we cannot use the same approach.  

To see how affine independence is useful, consider three affinely independent points $\{y_1,y_2,y_3\}$. Partition this set into two disjoint sets, $K_1, K_2$. We can draw a supporting hyperplane passing through any point from $\{y_1,y_2,y_3\}$ that separates $K_1$ and $K_2$ (see Figure~\ref{fig:supporting_hyperplanes}). In the proof, we consider affinely independent posterior beliefs. A general separation lemma then allows us to construct a hyperplane through $\posterior(\signal)$ to separate the set of posteriors below $\sighat$, from the set of posteriors above $\sighat$ (i.e., in $\sighup$ other than $\signal$). Formally, we  can define  $\a: \states \to \real$ such that, \begin{align*}
    \max_{\signal' < \sighat} \E_{\posterior(\signal')}[\a(\stater)] < \E_{\posterior(\signal)} [\a(\stater)] < \min_{\signal' \in \sighup, \signal' \neq \signal}\E_{\posterior(\signal')}[\a(\stater)].
\end{align*}
Using $\a(\cdot)$ defined above, and $\beta(\state') = \ind_{\state' =\state}$ we can construct a game $\widetilde\Gamma$ such that $\strategy(\signal') = \ind_{\signal' \ge \sighat}$ is an equilibrium in $\G$ under $\distq$ but not under $\dist$ leading to a desired contradiction. This argument also suggests that affine independence may not be the weakest requirement on the set of posteriors. 

\begin{figure}[h!]
\centering

\begin{minipage}[t]{0.3\textwidth} 
\centering
\begin{tikzpicture}[scale=1.5] 
    \fill (0,0) circle (2pt) node[below left] {$y_1$};
    \fill (1,0) circle (2pt) node[below] {$y_2$};
    \fill (0,1) circle (2pt) node[above right] {$y_3$};
    \draw[thick] (-0.2,-0.2) -- (1.2,1.2);
\end{tikzpicture}
\caption*{Supporting hyperplane through $y_1$ separating $\{y_1,y_2\}$ from $\{y_3\}$.}
\end{minipage}
\hfill
\begin{minipage}[t]{0.3\textwidth} 
\centering
\begin{tikzpicture}[scale=1.5] 
    \fill (0,0) circle (2pt) node[below left] {$y_1$};
    \fill (1,0) circle (2pt) node[below] {$y_2$};
    \fill (0,1) circle (2pt) node[above right] {$y_3$};
    \draw[thick] (-0.4,0.7) -- (1.4,-0.2);
\end{tikzpicture}
\caption*{Supporting hyperplane through $y_2$ separating $\{y_2,y_3\}$ from $\{y_1\}$.}
\end{minipage}
\hfill
\begin{minipage}[t]{0.3\textwidth} 
\centering
\begin{tikzpicture}[scale=1.5] 
    \fill (0,0) circle (2pt) node[below left] {$s_1$};
    \fill (1,0) circle (2pt) node[below] {$s_2$};
    \fill (0,1) circle (2pt) node[above right] {$s_3$};
    \draw[thick] (-0.2,1.4) -- (0.7,-0.4);
\end{tikzpicture}
\caption*{Supporting hyperplane through $y_3$ separating $\{y_1,y_3\}$ from $\{y_2\}$.}
\end{minipage}

\caption{Illustration of supporting hyperplanes separating sets of points.}
\label{fig:supporting_hyperplanes}
\end{figure}

A commonly studied special case of this setting is the separable environment, in which $\beta$ is a state-independent constant. In the appendix \ref{subsection: separable}, we show that Theorem~\ref{Theorem: CCAD with state} continues to hold in the separable environment. However, the difference is that we do not need the signals to be more similar in each state (i.e., $\dist^\state\ccad \distq^\state$), but only in expectation.

\begin{remark}
\label{remark:commonbelief}
A large global games literature following \cite{morris2002social} has explored the impact of a new public signal on coordination. The main difference is that unlike a new public signal, CAD similarity does not change an agent's belief about the underlying state since the marginal distribution of signals remain the same. However, a higher CAD similarity expands the set of signals where agents have common p-belief that $\state\in E$ for any $E\subset \Theta$. To see this, note that the set of types who p-believes $\state\in E$, i.e., $P(\state\in E|s)>p$, remains the same. Let’s call this $S^1$. However, after information becomes more CAD-similar, any type in $S^1$ believes that the other player is more likely to p-believe $\state\in E$. So, the set of types in $S^1$ who p-believe that the other player p-believes $\state\in E$ expands. Let’s call this $S^2$. Similarly, $S^3$ expands and so on. Since the common p-belief that $\state\in E$ is the intersection of all these sets, it also expands.   
\end{remark}

\section{Applications}

Theorems~\ref{Theorem: wCAD equivalence} and ~\ref{Theorem: CCAD with state} make the case for CAD orders by showing the equivalence of CAD-similarity and strategic similarity in canonical coordination games. This equivalence makes the CAD-orders particularly well-suited to compare information similarity in Bayesian games. In this section, we show how CAD orders can be used to derive economically meaningful predictions even in settings much beyond the games considered so far.

\subsection{Does more similar information improve welfare?}

We now revisit the bank run example introduced earlier. Our theorems demonstrate that greater similarity in information encourages agents to adopt the same strategy, which in turn increases the maximal run and decreases the minimal run across all equilibria. This raises the question of regulatory intervention: should a regulator increase or decrease information similarity to prevent bank runs? The appropriate intervention depends on equilibrium selection. 

For instance, if the regulator anticipates the agents will play the adversarial (maximal run) equilibrium then she prefers lesser homogenization of information. 
To see this, consider the example in the introduction. There is an equilibrium in which a player only stays when she sees signal $\frac{3}{2}$ if 
\begin{equation*}
\sum_\theta P \left(\theta\middle |s_i=\frac{1}{2}\right) \underbrace{\left(\theta+P\left(s_j=\frac{3}{2}\middle|s_i=\frac{1}{2},\theta\right)-1\right)}_{\text{Expected payoff from staying in state~$\theta$}}\leq 0.
\end{equation*}
The above inequality implies that a player who sees signal $\signal_i=\frac{1}{2}$ prefers running if the other player only stays when she sees $\signal_j=\frac{3}{2}$. Consider an information structure $\distq$ in which the above inequality is not satisfied. However, suppose information becomes more cCAD similar in state $\state=\frac{1}{2}$ and remains unchanged in other states. Then, a player who receives signal $\signal_i=\frac{1}{2}$ assigns a lower probability that the other player sees a signal $\signal_j=\frac{3}{2}$ if the state is $\state=\frac{1}{2}$, which decreases her incentive to stay. A sufficient decrease can satisfy the above inequality, thus creating an equilibrium in which a player only stays when she sees signal $\frac{3}{2}$. Thus, expected number of players who will run on the bank under the adversarial equilibrium has increased after the information becomes more cCAD similar. 
Thus, if the regulator anticipates that agent will play the adversarial equilibrium, then in the above example, the expected number of players who run on the bank increases after the information becomes more cCAD similar, making lesser homogenization of information desirable.

Conversely, if the social planner anticipates the agents will play the advantageous (minimal run) equilibrium then she prefers greater homogenization of information. To see this, consider the equilibrium in which a player only runs when she sees signal $-\frac{1}{2}$. This is an equilibrium if 
\begin{equation*}
\sum_\theta P \left(\theta\middle |s_i=\frac{1}{2}\right) \underbrace{\left(\theta+P\left(s_j\in\{\frac{1}{2},\frac{3}{2}\}\middle|s_i=\frac{1}{2},\theta\right)-1\right)}_{\text{Expected payoff from staying in state~$\theta$}}\geq 0.
\end{equation*}
The above inequality implies that a player who sees signal $\signal_i=\frac{1}{2}$ prefers staying if the other player stays when she sees $\signal_j\in\{\frac{1}{2},\frac{3}{2}\}$. Consider an information structure $\distq$ in which the above inequality is not satisfied. However, suppose information becomes more cCAD similar in state $\state=\frac{1}{2}$ and remains unchanged in other states. Then, a player who receives signal $\signal_i=\frac{1}{2}$ assigns a higher probability that the other player sees a signal $\signal_j\in \{\frac{1}{2},\frac{3}{2}\}$ if the state is $\state=\frac{1}{2}$, which increases her incentive to stay. A sufficient increase can satisfy the above inequality, thus creating an equilibrium in which a player only runs when she sees signal $-\frac{1}{2}$. Thus, if the regulator anticipates that agent will play the advantageous equilibrium, then in the above example, the expected number of players who run on the bank decreases after the information becomes more cCAD similar, making greater homogenization of information desirable.

\subsection{Beyond affine coordination games}

Next, we move from affine and increasing $h(\cdot)$ to simply increasing $h(\cdot)$. We define a stronger notion of similarity that we call \emph{Strong Concentration Along the Diagonal (sCAD)}.  We say information structure~$\dist$ is more similar than $\distq$, or is greater in the strong CAD order, if any agent $i$ believes, conditional on her realized signal $s$, that the number of other agents with the a similar signal (signal in any set $K$ that contains $\signal$) is higher in the sense of \emph{first-order stochastic dominance} under $\dist$ than $\distq$. Under CAD, this ranking was only in terms of expectation.  We show, analogous to our first result that increasing similarity of information in the sCAD order is equivalent to expanding the set of equilibria in the class of private-value coordination games for any increasing $h(\cdot)$.

\bthm\label{Theorem: sCAD equivalence} For any two joint distributions
$\dist$ and $\distq$, $\dist \scad \distq$ if and only if $\eq(\G, \dist) \supseteq \eq(\G,\distq)$ for all private-value coordination games~$\G$.

\ethm

\subsection{Similarity and Congestion}
Consider a binary-action game of strategic substitutability in which $N$ players simultaneously choose between actions~$0$ or $1$, and the payoff difference 
\begin{align}
    \payoff(\ractoth,\signal_i):=u(a_i=1,\ractoth,\signal_i)-u(a_i=0,\ractoth,\signal_i)=\a(\signal_i)+\b(\signal_i)h(\ractoth)\label{Equation: payoff difference congestion games}
\end{align}
decreases in the aggregate action, that is, $\beta(.)<0$. Since $\beta(\cdot)< 0$ for all $s\in \sigset$, the game~$\G$ exhibits strategic substitutability. Henceforth, we refer to such symmetric binary-action games of strategic substitutability simply as \textit{congestion games}.  In contrast to coordination games, a player in a congestion game has a \textit{smaller} incentive to play action $1$ if more of the other players play action $1$.
\begin{definition}[\textbf{(Affine) Congestion Games}]
A game~$\hat \G$ a called a \textbf{congestion game} when the payoff difference can be described by Equation \eqref{Equation: payoff difference congestion games}, with  increasing $h(\cdot)$ and negative $\b(\cdot)$. The congestion game~$\hat \G$ is said to be \textbf{affine} if $h(\cdot)$ is affine.
\end{definition}
 A classic example is an entry game in a market with capacity constraints. A firm suffers when more firms enter making the market more congested \cite[see][]{duffy2005learning}. The conflict game of \cite{baliga2012strategy} is another interesting example in which a country wants to be aggressive whenever its opponent is peaceful and vice versa. A public good contribution game in which an agent's contribution to the public good matters less for successful public good provision when more agents contribute also exhibits strategic substitutability \cite[see][]{harrison2021global}.
 
 We can establish equivalence results that are analogous to our main results: More CAD-similar information shrinks the set of symmetric pure strategy BNE in congestion games. 
 
\begin{proposition}
    \label{proposition:congestion}
Let $\dist$ and $\distq$ be two joint distributions over $\sigset^N$. 
\begin{enumerate}
    \item $\dist \cad \distq$ if and only if $\eq(\G, \dist) \subseteq \eq(\G,\distq)$ for all affine congestion games $\hat \G$. 
    \item $\dist \scad \distq$ if and only if $\eq(\G, \dist) \subseteq \eq(\G,\distq)$ for all congestion games $\hat \G$.
\end{enumerate}
\end{proposition}

\subsection{Similarity and Collective Action}\label{sec:change}
Recent public discourse and research in political science suggests that social media or access to the same information has enabled larger mass protests.\footnote{See \citet{manacorda2020liberation}, \citet{qin2024social} for example.} Protests are fundamentally collective action problems: Protesting is costly and a successful regime change requires a sufficient number of people to take this costly action, while the benefit of regime change accrues to all. So, while people want to coordinate to ensure a successful mass protest, they are also tempted to free-ride. Importantly, unlike the games we studied so far in this paper, collective action games do not exhibit strategic complementarity. Moreover, equilibrium may not be symmetric. In a different paper~\cite{basak2024protest}, we show that the notion of CAD can still be used to characterize when increased information similarity helps or harms participation in a collective action game of incomplete information. The main insight is that more similar information about the fundamentals in the sense of increased CAD is a double-edged sword: It can help agents coordinate, but can also exacerbate free-riding. We show that more similar information facilitates (impedes) collective action when achieving regime change is sufficiently challenging (easy). 

\subsection{Similarity and  Auctions}\label{Section: second price auction}
A classical question in auction theory (e.g., \cite{milgrom1982theory}) is how different auction formats compare in terms of revenue, when players' valuations are more interdependent. We can use the tools developed in this paper to answer a related question: Given an auction format, how does revenue change if player valuations become more similar? For example, below, we show that the revenue from a second-price auction unambiguously increases when players' valuations become more similar in the CAD order.

Consider a second-price auction for a single object with two bidders. Bidder valuations are drawn from a finite set $\sigset$. Let $\dist$ and $\distq$ be two different joint distributions over the bidder valuations. Let $R(\cdot)$ denote the expected revenue from the second-price auction given a joint distribution, when bidders report their valuations truthfully. Truthful reporting is weakly dominant strategy in this environment. The result below shows that when bidder valuations are more similar in the sense of an increase in the CAD order, then the expected revenue in the second-price auction is higher. The proof is in the appendix. 

\bprop\label{Proposition: second price auction revenue} If $\dist \cad \distq$, then $R(\dist) \ge R(\distq)$. \eprop

\subsection{Games with non-exchangeable signal distributions}
Our baseline setting features a lot of symmetry: Player identities are not payoff-relevant and the joint distribution over agents' signals is exchangeable. In Appendix \ref{subsec:nonexchangeability}, we relax these symmetry assumptions. We consider a class of binary-action coordination games in which payoffs depend on a weighted aggregate action by others (with identity-specific weights) and allow for non-exchangeable signal distributions. We define an analogous notion of  CAD for non-exchangeable joint distributions, and derive a characterization analogous to Theorem~\ref{Theorem: wCAD equivalence}.

\bibliography{ref_orders.bib}

@article{farrell2007coordination,
  title={Coordination and lock-in: Competition with switching costs and network effects},
  author={Farrell, Joseph and Klemperer, Paul},
  journal={Handbook of industrial organization},
  volume={3},
  pages={1967--2072},
  year={2007},
  publisher={Elsevier}
}

@article{farrell1986installed,
  title={Installed base and compatibility: Innovation, product preannouncements, and predation},
  author={Farrell, Joseph and Saloner, Garth},
  journal={The American economic review},
  pages={940--955},
  year={1986},
  publisher={JSTOR}
}

@article{sakovics2012matters,
  title={Who matters in coordination problems?},
  author={Sakovics, Jozsef and Steiner, Jakub},
  journal={American Economic Review},
  volume={102},
  number={7},
  pages={3439--3461},
  year={2012},
  publisher={American Economic Association}
}

@inbook{Morris_Shin_2003,
  author       = {Morris, Stephen and Shin, Hyun Song},
  title        = {Global Games: Theory and Applications},
  booktitle    = {Advances in Economics and Econometrics: Theory and Applications, Eighth World Congress},
  editor       = {Dewatripont, Mathias and Hansen, Lars Peter and Turnovsky, Stephen J.},
  series       = {Econometric Society Monographs},
  volume       = {1},
  year         = {2003},
  pages        = {56--114},
  publisher    = {Cambridge University Press},
  address      = {Cambridge}
}

@article{morris2004coordination,
  title={Coordination risk and the price of debt},
  author={Morris, Stephen and Shin, Hyun Song},
  journal={European Economic Review},
  volume={48},
  number={1},
  pages={133--153},
  year={2004},
  publisher={Elsevier}
}

@article{morris2016common,
  title={Common belief foundations of global games},
  author={Morris, Stephen and Shin, Hyun Song and Yildiz, Muhamet},
  journal={Journal of Economic Theory},
  volume={163},
  pages={826--848},
  year={2016},
  publisher={Elsevier}
}

@article{qin2024social,
  title={Social media and collective action in China},
  author={Qin, Bei and Str{\"o}mberg, David and Wu, Yanhui},
  journal={Econometrica},
  volume={92},
  number={6},
  pages={1993--2026},
  year={2024},
  publisher={Wiley Online Library}
}

@article{basak2024protest,
  title={Similarity of information and collective action},
  author={Basak, Deepal and Deb, Joyee and Kuvalekar, Aditya},
  year={Forthcoming},
  journal = {The American Economic Review},
  publisher={}
}

@article{milgrom1982theory,
  title={A theory of auctions and competitive bidding},
  author={Milgrom, Paul R and Weber, Robert J},
  journal={Econometrica: Journal of the Econometric Society},
  pages={1089--1122},
  year={1982},
  publisher={JSTOR}
}

@inproceedings{anwar2024filter,
  title={Filter Bubble or Homogenization? Disentangling the Long-Term Effects of Recommendations on User Consumption Patterns},
  author={Anwar, Md Sanzeed and Schoenebeck, Grant and Dhillon, Paramveer S},
  booktitle={Proceedings of the ACM on Web Conference 2024},
  pages={123--134},
  year={2024}
}

@inproceedings{chaney2018algorithmic,
  title={How algorithmic confounding in recommendation systems increases homogeneity and decreases utility},
  author={Chaney, Allison JB and Stewart, Brandon M and Engelhardt, Barbara E},
  booktitle={Proceedings of the 12th ACM conference on recommender systems},
  pages={224--232},
  year={2018}
}

@inproceedings{aridor2020deconstructing,
  title={Deconstructing the filter bubble: User decision-making and recommender systems},
  author={Aridor, Guy and Goncalves, Duarte and Sikdar, Shan},
  booktitle={Proceedings of the 14th ACM conference on recommender systems},
  pages={82--91},
  year={2020}
}

@inproceedings{nguyen2014exploring,
  title={Exploring the filter bubble: the effect of using recommender systems on content diversity},
  author={Nguyen, Tien T and Hui, Pik-Mai and Harper, F Maxwell and Terveen, Loren and Konstan, Joseph A},
  booktitle={Proceedings of the 23rd international conference on World wide web},
  pages={677--686},
  year={2014}
}

@techreport{CB24,
  title={Diversity, Disagreement, and Information Aggregation},
  author={Cheng, Xienan and Borgers, Tilman},
  year={2024},
  institution={Working paper}
}

@techreport{de2023robust,
  title={Robust Aggregation of Correlated Information},
  author={de Oliveira, Henrique and Ishii, Yuhta and Lin, Xiao},
  year={2023},
  institution={Working Paper}
}

@article{clemen1985limits,
  title={Limits for the precision and value of information from dependent sources},
  author={Clemen, Robert T and Winkler, Robert L},
  journal={Operations Research},
  volume={33},
  number={2},
  pages={427--442},
  year={1985},
  publisher={INFORMS}
}

@article{jensen2018distributional,
  title={Distributional comparative statics},
  author={Jensen, Martin Kaae},
  journal={The Review of Economic Studies},
  volume={85},
  number={1},
  pages={581--610},
  year={2018},
  publisher={Oxford University Press}
}

@article{mekonnen2022bayesian,
  title={Bayesian comparative statics},
  author={Mekonnen, Teddy and Vizca{\'\i}no, Ren{\'e} Leal},
  journal={Theoretical Economics},
  volume={17},
  number={1},
  pages={219--251},
  year={2022},
  publisher={Wiley Online Library}
}

@article{bergemann2013robust,
  title={Robust predictions in games with incomplete information},
  author={Bergemann, Dirk and Morris, Stephen},
  journal={Econometrica},
  volume={81},
  number={4},
  pages={1251--1308},
  year={2013},
  publisher={Wiley Online Library}
}

@article{angeletos2007efficient,
  title={Efficient use of information and social value of information},
  author={Angeletos, George-Marios and Pavan, Alessandro},
  journal={Econometrica},
  volume={75},
  number={4},
  pages={1103--1142},
  year={2007},
  publisher={Wiley Online Library}
}

@article{H71,
 author = {Jack Hirshleifer},
 journal = {The American Economic Review},
 number = {4},
 pages = {561--574},
 publisher = {American Economic Association},
 title = {The Private and Social Value of Information and the Reward to Inventive Activity},
 urldate = {2023-03-20},
 volume = {61},
 year = {1971}
}

@article{cookson2023social,
  title={Social media as a bank run catalyst},
  author={Cookson, J Anthony and Fox, Corbin and Gil-Bazo, Javier and Imbet, Juan Felipe and Schiller, Christoph},
  journal={Available at SSRN},
  volume={4422754},
  year={2023},
  publisher={Working Paper, University of Colorado at Boulder}
}

@article{gam2023does,
  title={Does Social Media Make Banks More Fragile? Evidence from Twitter},
  author={Gam, Yong Kyu and Liu, Chunbo and Xu, Yongxin},
  journal={Evidence from Twitter (November 18, 2023)},
  year={2023}
}

@article{manacorda2020liberation,
  title={Liberation technology: Mobile phones and political mobilization in Africa},
  author={Manacorda, Marco and Tesei, Andrea},
  journal={Econometrica},
  volume={88},
  number={2},
  pages={533--567},
  year={2020},
  publisher={Wiley Online Library}
}

@article{morris2002social,
  title={Social value of public information},
  author={Morris, Stephen and Shin, Hyun Song},
  journal={The American economic review},
  volume={92},
  number={5},
  pages={1521--1534},
  year={2002},
  publisher={American Economic Association}
}

@article{nechushtai2024more,
  title={More of the same? Homogenization in news recommendations when users search on Google, YouTube, Facebook, and Twitter},
  author={Nechushtai, Efrat and Zamith, Rodrigo and Lewis, Seth C},
  journal={Mass Communication and Society},
  volume={27},
  number={6},
  pages={1309--1335},
  year={2024},
  publisher={Taylor \& Francis}
}

@article{muller2000some,
  title={Some remarks on the supermodular order},
  author={M{\"u}ller, Alfred and Scarsini, Marco},
  journal={Journal of multivariate analysis},
  volume={73},
  number={1},
  pages={107--119},
  year={2000},
  publisher={Elsevier}
}

@book{meyer1990interdependence,
  title={Interdependence in multivariate distributions: stochastic dominance theorems and an application to the measurement of ex post inequality under uncertainty},
  author={Meyer, Margaret A},
  year={1990},
  publisher={Nuffield College}
}

@article{awaya2022common,
  title={Commonality of Information and Commonality of Beliefs},
  author={Awaya, Yu and Krishna, Vijay},
  journal={Theoretical Economics},
  volume={Forthcoming},
  number={},
  pages={},
  year={2025},
  publisher={}
}

@article{meyer2015beyond,
  title={Beyond correlation: Measuring interdependence through complementarities},
  author={Meyer, Margaret and Strulovici, Bruno},
  journal = {Working paper},
  year={2015},
  publisher={University of Oxford}
}

@article{shadmehr2011collective,
  title={Collective action with uncertain payoffs: coordination, public signals, and punishment dilemmas},
  author={Shadmehr, Mehdi and Bernhardt, Dan},
  journal={American Political Science Review},
  volume={105},
  number={4},
  pages={829--851},
  year={2011},
  publisher={Cambridge University Press}
}

@article{gossner2000comparison,
  title={Comparison of information structures},
  author={Gossner, Olivier},
  journal={Games and Economic Behavior},
  volume={30},
  number={1},
  pages={44--63},
  year={2000},
  publisher={Elsevier}
}

@article{cherry2012strategically,
  title={Strategically valuable information},
  author={Cherry, Josh and Smith, Lones},
  journal={Available at SSRN 2722234},
  year={2012}
}

@article{morris1998unique,
  title={Unique equilibrium in a model of self-fulfilling currency attacks},
  author={Morris, Stephen and Shin, Hyun Song},
  journal={American Economic Review},
  pages={587--597},
  year={1998},
  publisher={JSTOR}
}

@article{bergemann2016bayes,
	Author = {Bergemann, Dirk and Morris, Stephen},
	Journal = {Theoretical Economics},
	Number = {2},
	Pages = {487--522},
	Publisher = {Wiley Online Library},
	Title = {Bayes correlated equilibrium and the comparison of information structures in games},
	Volume = {11},
	Year = {2016}}

@book{muller2002comparison,
  title={Comparison methods for stochastic models and risks},
  author={M{\"u}ller, Alfred and Stoyan, Dietrich},
  volume={389},
  year={2002},
  publisher={Wiley}
}

@article{cheng2022reporting,
  title={Reporting sexual misconduct in the\# MeToo era},
  author={Cheng, Ing-Haw and Hsiaw, Alice},
  journal={American Economic Journal: Microeconomics},
  volume={14},
  number={4},
  pages={761--803},
  year={2022}
}

@article{konrad2016coordination,
  title={Coordination and the fight against tax havens},
  author={Konrad, Kai A and Stolper, Tim BM},
  journal={Journal of International Economics},
  volume={103},
  pages={96--107},
  year={2016},
  publisher={Elsevier}
}

@article{vives2014strategic,
  title={Strategic complementarity, fragility, and regulation},
  author={Vives, Xavier},
  journal={The Review of Financial Studies},
  volume={27},
  number={12},
  pages={3547--3592},
  year={2014},
  publisher={Oxford University Press}
}

@article{goldstein2005demand,
  title={Demand--deposit contracts and the probability of bank runs},
  author={Goldstein, Itay and Pauzner, Ady},
  journal={the Journal of Finance},
  volume={60},
  number={3},
  pages={1293--1327},
  year={2005},
  publisher={Wiley Online Library}
}

@article{meyer2012increasing,
  title={Increasing interdependence of multivariate distributions},
  author={Meyer, Margaret and Strulovici, Bruno},
  journal={Journal of Economic Theory},
  volume={147},
  number={4},
  pages={1460--1489},
  year={2012},
  publisher={Elsevier}
}

@article{cripps2008common,
  title={Common learning},
  author={Cripps, Martin W and Ely, Jeffrey C and Mailath, George J and Samuelson, Larry},
  journal={Econometrica},
  volume={76},
  number={4},
  pages={909--933},
  year={2008},
  publisher={Wiley Online Library}
}

@article{epstein1980increasing,
  title={Increasing generalized correlation: a definition and some economic consequences},
  author={Epstein, Larry G and Tanny, Stephen M},
  journal={Canadian Journal of Economics},
  pages={16--34},
  year={1980},
  publisher={JSTOR}
}

@article{harrison2021global,
  title={Global games with strategic substitutes},
  author={Harrison, Rodrigo and Jara-Moroni, Pedro},
  journal={International Economic Review},
  volume={62},
  number={1},
  pages={141--173},
  year={2021},
  publisher={Wiley Online Library}
}

@article{baliga2012strategy,
  title={The strategy of manipulating conflict},
  author={Baliga, Sandeep and Sj{\"o}str{\"o}m, Tomas},
  journal={American Economic Review},
  volume={102},
  number={6},
  pages={2897--2922},
  year={2012},
  publisher={American Economic Association}
}

@article{duffy2005learning,
  title={Learning, information, and sorting in market entry games: theory and evidence},
  author={Duffy, John and Hopkins, Ed},
  journal={Games and Economic behavior},
  volume={51},
  number={1},
  pages={31--62},
  year={2005},
  publisher={Elsevier}
}

@article{katz1985network,
  title={Network externalities, competition, and compatibility},
  author={Katz, Michael L and Shapiro, Carl},
  journal={The American economic review},
  volume={75},
  number={3},
  pages={424--440},
  year={1985},
  publisher={JSTOR}
}

@book{joe1997multivariate,
  title={Multivariate models and multivariate dependence concepts},
  author={Joe, Harry},
  year={1997},
  publisher={CRC press}
}

\appendix

\section{Appendix: Proofs}
\subsection{Proof of Lemma~\ref{Lemma: wCAD equivalent with increasing expectation}}
\bprf 
First, we prove that 1. $\iff$ 2. Suppose that $\dist \cad \distq$ and for some $K$ and $\signal \in K$ and $i,j$, we have $Prob(\sigr_j \in K \vert \sigr_i = \signal) < Prob(\sigry_j \in K \vert \sigry_i = \signal)$. Then, at least for some $\signal' \notin K$, $Prob(\sigr_j = \signal' \vert \sigr_i = \signal) < Prob(\sigry_j = \signal' \vert \sigry_i = \signal)$, contradicting $\dist \cad \distq$. For the converse, suppose that the inequality holds for all $i,j$, $\signal$ and $K \ni \signal$, but $\dist \ncad \distq$. Then, for some $\signal'\neq \signal$, $Prob(\sigr_j = \signal' \vert \sigr_i = \signal) > Prob(\sigry_j = \signal' \vert \sigry_i = \signal)$. But then, $K = \sigset \backslash \{\signal'\}$ would have $$Prob(\sigr_j \in K \vert \sigr_i = \signal) > Prob(\sigry_j \in K \vert \sigry_i = \signal), $$ a contradiction. 

Next, we prove 1. $\iff$ 3. $\dist \cad \distq$ 
\begin{align*}
     &\Longleftrightarrow Prob(\sigr_j \in K \bigg\vert \sigr_i = \signal) \ge Prob(\sigry_j \in K \bigg\vert \sigry_i = \signal) \quad \forall \signal, K \ni \signal, \quad 
    \\
    & \Longleftrightarrow \E\left[ \ind_{\sigr_j \in K} \bigg\vert \sigr_i = \signal\right] \ge \E\left[ \ind_{\sigry_j \in K} \bigg\vert \sigry_i = \signal\right] \quad \forall \signal, K \ni \signal. \\
    & \Longleftrightarrow (\nplayers-1) \E\left[ \ind_{\sigr_j \in K} \bigg\vert \sigr_i = \signal\right] \ge (\nplayers-1) \E\left[ \ind_{\sigry_j \in K} \bigg\vert \sigry_i = \signal\right] \quad \forall \signal, K \ni \signal. \\
    & \Longleftrightarrow \E\left[ \sum_{j\neq i} \ind_{\sigr_j \in K} \bigg\vert \sigr_i = \signal\right] \ge  \E\left[ \sum_{j\neq i} \ind_{\sigry_j \in K} \bigg\vert \sigry_i = \signal\right] \quad \forall \signal, K \ni \signal. 
\end{align*}
\eprf

\subsection{Proof of Theorem~\ref{Theorem: wCAD equivalence}}

\textbf{Step 1: $(1) \implies (2)$} 

Suppose that $\dist \cad \distq$. Let $\strategy \in \eq(\G,\distq)$ for some game $\G$ with $\b(\signal) \ge 0$ for all $\signal \in \sigset$. We show that~$\strategy$  remains an equilibrium under the more CAD similar information structure~$\dist$, i.e., $\strategy \in \eq(\G,\dist)$.

Since $\strategy$ constitutes an equilibrium under $\distq$, (\ref{Equation: IC to play a=1}) must hold for any $\signal \in P(\sigma)$, and  (\ref{Equation: IC to play a=0}) must hold for any $\signal \in NP(\sigma)$. Fix an agent~$i$. We compare the net payoffs from taking action~$a=1$ under $\dist$ and $\distq$. 
\begin{align*}
&\E[ \payoff(\ractoth, \signal) \vert \signal,\strategy;\dist ] - \E[ \payoff(\ractoth, \signal) \vert \signal,\strategy;\distq ] \\
=&\b(s) \left[ \E[ h(\ractoth) \vert \signal,\strategy;\dist ] - \E[h(\ractoth) \vert \signal,\strategy;\distq ] \right].
\end{align*}
Since $h(\cdot)$ is affine and increasing,  $h(y) = k y + l$ for some $k > 0$. So we can rewrite the above as follows: 

\begin{align*}
=&k\b(s) \left( \E\left[\ractoth\bigg\vert \signal,\strategy;\dist\right] - \E\left[ \ractoth\bigg\vert \signal,\strategy;\distq\right]\right)\\
=&k \b(s) \left( \E\left[ \sum_{j\neq i}\ind_{\sigrs_j \in P(\strategy)} \bigg\vert \signal;\dist\right] -\E\left[ \sum_{j\neq i}\ind_{\sigrs_j \in P(\strategy)} \bigg\vert \signal;\distq\right] \right)\\
=& \b(\signal)  k \left[\E\left[ \countr( P(\strategy)) \bigg\vert \signal;\dist\right] -\E\left[ \countr( P(\strategy)) \bigg\vert \signal;\distq\right]\right],
\end{align*}
where, recall that $\countr(P(\strategy))$ counts the number of players other than player $i$ with signal in set $P(\strategy)$. 
By Lemma~\ref{Lemma: wCAD equivalent with increasing expectation} and since $\b(\signal) \ge 0 \forall \signal \in \sigset$, the above expression is non-negative for $\signal \in P(\strategy)$ and non-positive for $\signal \in NP(\strategy)$. Therefore, for strategy profile $\strategy$, we have 
\begin{align}\label{complementarityproof}
 \E[ \payoff(\ractoth, \signal) \vert \signal,\strategy;\dist ] - \E[ \payoff(\ractoth, \signal) \vert \signal,\strategy;\distq ] &\ge  0 \quad \text{ if } \signal \in P(\strategy)\notag \\
 &\le  0 \quad \text{ if } \signal \in NP(\strategy) 
\end{align}
Therefore, \eqref{Equation: IC to play a=1} continues to hold for all $\signal \in P(\strategy)$ under $\dist$ and \eqref{Equation: IC to play a=0} holds for all $\signal \in NP(\strategy)$ under $\dist$.

\medskip
 
\textbf{Step 2: $(2) \implies (3), (4)$} 

Since $\dist$ and $\distq$ have the same marginals, and for any $\G$, $\eq(\G, \dist) \supseteq \eq(\G,\distq)$,
$$\maxP(\G, \dist)=\max_{\strategy \in \eq(\G,\dist)} \dist(P(\strategy)) \ge \max_{\strategy \in \eq(\G,\distq)} \distq(P(\strategy))=\maxP(\G,\distq)$$ 
$$\minP(\G, \dist)= \min_{\strategy \in \eq(\G,\dist)} \dist(P(\strategy)) \le \min_{\strategy \in \eq(\G,\distq)} \distq(P(\strategy))= \minP(\G,\distq).$$

\medskip

\textbf{Step 3: $\mathbf{(3)}\;\boldsymbol{\implies}\;\mathbf{(1)}$}

Suppose that $\maxP(\G, \dist) \ge \maxP(\G,\distq)$ for all $\G$ with $\b(\signal) \ge 0$ for all $\signal \in \sigset$, but $\dist \ncad \distq$. Therefore, $\exists \signal^*\in \sigset$ and $K \ni \signal^*$ such that, 
$$\dist(\sigrs_j \in K \vert \sigrs_i = \signal^*) 
< \distq(\sigrs_j \in K \vert \sigrs_i = \signal^*).$$
The proof approach will be to establish a contradiction by constructing a game $\G$, and a strategy profile $\strategy$ such that the following hold: 
\begin{enumerate}[(i)]
    \item $\strategy \in \eq(\G, \distq)$ with $P(\strategy) = K$,
    \item $ \maxP(\G,\distq) = \distq(P(\strategy)) > \distq(P(\strategy'))$ for all $\strategy' \in \eq(\G,\distq)$ with $\strategy'\neq \strategy$,
    \item $\strategy \notin \eq(\G,\dist)$, and
    \item no $\strategy' \in \eq(\G,\dist)$ such that $\dist(P(\strategy')) \ge \dist(P(\strategy))$.
\end{enumerate}
\noindent 

Consider a game $\G$ with:
\[
\begin{array}{ll}
 \a(\signal) &= \begin{cases} -2 & \text{ if } \signal \notin K\\
 -\distq_{\signal^*}(K) & \text{ if } \signal = \signal^*, \\
 1 & \text{ otherwise.} \end{cases}\\
 h(\actoth) &= \actoth, \\
 \b(\signal) &= \frac{1}{\nplayers-1}.
\end{array}
\]

Under these parameters, $\E[d(\ractoth,\signal) \vert \signal;\distq) = \a(\signal) + \b(\signal) \E[\ractoth\vert \signal;\distq] = \a(\signal) + \distq_\signal(K)$. Consider the strategy $\strategy$ with $P(\strategy)=K$. For all $\signal \in \sigset$,
\begin{align*}
   \E[d(\ractoth,\signal\vert \signal,\distq)]=\begin{cases} -2 + \distq_\signal(K) < 0 & \text{ if } \signal \notin K\\
    -\distq_{\signal^*}(K) +  \distq_{\signal^*}(K) = 0 & \text{ if } \signal = \signal^*\\
    1 + \distq_\signal(K) > 0 & \text{ if } \signal \in K, \signal \neq \signal^*.\end{cases}
\end{align*}
Therefore, under $\distq$, \eqref{Equation: IC to play a=1} and \eqref{Equation: IC to play a=0} are satisfied, so $\strategy \in \eq(\G,\distq).$ This gives us feature (i).

Furthermore, for any $\signal \notin K$,  regardless of the strategy others play, $$\E[d(\ractoth,\signal\vert \signal,\distq)]<0.$$ This means participation is a dominated strategy for these types, which makes $\sigma$ the maximal equilibrium, i.e., $\maxP(\G,\distq) = \distq(P(\strategy))$. This gives us feature (ii).

Recall that for type $\signal^*$,
$\dist(\sigrs_j \in K \vert \sigrs_i = \signal^*)
< \distq(\sigrs_j \in K \vert \sigrs_i = \signal^*).$
Therefore,
\begin{align*}
\E\left[ \payoff(\ractoth,\signal^*) \mid \signal^*; \dist\right]
    &=\dist_{\signal^*}( K )
- \distq_{\signal^*}(K) < 0,
\end{align*}
This means while type $\signal^*$ in indifferent under $\distq$, she strictly prefers not participation under $\dist$. Thus, $\strategy \notin \eq(\G, \dist)$. This gives us feature (iii).

Since participation is a dominated strategy for any $s\notin K$ (regardless of the information structure), there is no other equilibrium under $\dist$ which leads to greater participation that $\dist(P(\strategy))$.  This gives us feature (iv).

Finally, since $\dist$ and $\distq$ has the same marginal distribution, it follows from the above steps that $\maxP(\G,\dist)\leq \dist(P(\strategy))=\distq(P(\strategy))= \maxP(\G,\distq)$, which contradicts  $(3)$.
\bigskip

\textbf{Step 4: $\mathbf{(4)}\;\boldsymbol{\implies}\;\mathbf{(1)}$}

This argument is similar to Step 3.
If (1) is not true, then $\exists \signal^*\in \sigset$ and $K \ni \signal^*$ such that,
$\dist(\sigrs_j \in K \vert \sigrs_i = \signal^*)
< \distq(\sigrs_j \in K \vert \sigrs_i = \signal^*).$ Let $K^c:=\sigset \backslash K$. Then,
$$\dist(\sigrs_j \in K^c\vert \sigrs_i = \signal^*)
> \distq(\sigrs_j \in K^c \vert \sigrs_i = \signal^*).$$
The proof approach will be to establish a contradiction by constructing a game $\G$, and a strategy profile $\strategy$ such that the following hold:
\begin{enumerate}[(i)]
    \item $\strategy \in \eq(\G, \distq)$ with $P(\strategy) = K^c$,
    \item $ \minP(\G,\distq) = \distq(P(\strategy)) < \distq(P(\strategy'))$ for all $\strategy' \in \eq(\G,\distq)$ with $\strategy'\neq \strategy$,
    \item $\strategy \notin \eq(\G,\dist)$, and
    \item no $\strategy' \in \eq(\G,\dist)$ such that $\dist(P(\strategy')) \le \dist(P(\strategy))$.
\end{enumerate}
As before, consider the following game:
\begin{align*}
    \a(\signal) =&  \begin{cases} 1 & \text{ if } \signal \in K^c \\
    -\distq_{\signal^*}(K^c) & \text{ if } \signal = \signal^*\\
    -2 & \text{ otherwise }
    \end{cases}
    \\
    \b(\signal) = & \frac{1}{\nplayers-1}\\
    h(\actoth) = & \actoth.
\end{align*}

Let $\strategy$ be the strategy profile with participation set $P(\strategy)=K^c$. Then,
\begin{align*}
    \E[d(\ractoth,\signal \vert \signal, \distq)] = \begin{cases}
        1 + \distq_\signal(K^c) > 0 & \text{ if } \signal \in K^c\\
        -\distq_{\signal^*}(K^c) + \distq_{\signal^*}(K^c) = 0 & \text{ if } \signal = \signal^* \\
        -2 + \distq_\signal(K^c) < 0 & \text{ if } \signal \notin K^c, \signal \neq \signal^*.
    \end{cases}
\end{align*}
Therefore, $\strategy \in \eq(\G,\distq).$ This gives us feature (i).

Note that not participation is a dominated strategy for any type $s\in K^c$ since  regardless of what others play, $\E[d(\ractoth,\signal\vert \signal;\distq)] \ge \a(\signal) > 0$. This makes $\strategy$ the minimal participation equilibrium.  This gives us feature (ii).

However, $$\E[d(\ractoth,\signal^* \vert \signal^*,\dist)] = -\distq_{\signal^*}(K^c) + \dist_{\signal^*}(K^c) > 0.$$
This means while type $\signal^*$ was indifferent under $\distq$, she strictly prefers participation under $\dist$. Therefore, $\sigma$ is no longer an equilibrium under $\dist$. This gives us feature (iii).

Since, regardless of the information structure, any type $s\in K^c$ always participates, the minimal participation under $\dist$ must be at least $\dist(P(\sigma))$. This gives us feature (iv).

Finally, since $\dist$ and $\distq$ has the same marginals, it follows from the above steps that and $\minP(\G,\dist) \geq \dist(P(\sigma))=\distq(P(\sigma))= \minP (\G,\distq)$, which contradicts (4).

\subsection{Proof of Theorem~\ref{Theorem: CCAD with state}}

We first establish a useful property of affine independence.
\blemma\label{Lemma: affine independence and separation} Let $x \in \real^n$ and $K, L,  \subset \real^n$ be two finite, disjoint sets such that, $x \notin K \bigcup L$ and $K \bigcup L\bigcup\{x\}$  is affinely independent. Then, $\exists \lmv \in \real^n$ such that $$\max_{k \in K} \lm' k < \lm' x < \min_{l \in L} \lm' l.$$
\elemma

\bprf
Since $K \bigcup L \bigcup\{ x\}$ is affinely independent, $\{k-x: k \in K\} \bigcup \{ l-x: l \in L\}$ is linearly independent. Complete this set to a basis, $\{ x_1, \ldots, x_n\}$ and define a linear map as follows:
\begin{align*}
    T(x_i) := \begin{cases} -1 & \text{ if } x_i = k -x \text{ for some } k \in K\\
    1 & \text{ if } x_i = l - x \text{ for some } l \in L\\
    0 & \text{ if } x_i \in T \backslash \left\{ \{k-x: k \in K\} \bigcup \{ l-x: l \in L\}\right\}.
    \end{cases}
\end{align*}
$T(y)$ for any $y$ is defined using the above basis vectors. By definition, $T(k-x) = T(k) - T(x) = -1 \implies T(k) < T(x) $ for all $k \in K$ and $T(l -x) = T(l) - T(x) = 1 \implies T(l) > T(x)$ for all $l \in L$. Finally, by the Riesz representation theorem, $\exists \lm \in \real^n$ such that $T(y) = \lm' y$ for all $y \in \real^n$.
\eprf

\bprf[Proof of Theorem~\ref{Theorem: CCAD with state}] To prove $(1)$, define $T \subseteq \states$ with $\dist = \distq^\state$ for all $\state \in \states\backslash T$, and suppose that $\dist^\state \ccad \distq^\state$ for all $\state \in T$. Consider a common-value affine coordination game $\G$. Affine, increasing  $h(\cdot)$ means $h(x) = k x + l$ for some $k \ge 0$. Suppose that $\strategy$ is a cutoff equilibrium under $\distq$ with a cutoff $\sigt$. Therefore,
\begin{align*}
    \E\left[\payoff(\ractoth, \stater) \vert \signal,\strategy;\distq\right] \ge 0 \quad \forall \signal \ge \tilde \signal\\
    \E\left[\payoff(\ractoth, \stater) \vert \signal,\strategy;\distq\right] \le 0 \quad \forall \signal < \tilde \signal
\end{align*}
For any $\signal$,
\begin{align*}
&\E\left[ \payoff(\ractoth, \stater) \vert \signal;\strategy,\dist \right] - \E\left[ \payoff(\ractoth, \stater) \vert \signal;\strategy,\distq \right]
\\
& =\E\left[ \a(\stater)+\b(\stater) h(\ractoth)\vert \signal;\strategy,\dist \right] - \left[\a(\stater)+\b(\stater) h(\ractoth) \vert \signal;\strategy,\distq \right]\\
& =\E\left[\b(\stater) k \E\left[ \ractoth\vert \stater, \dist^{\stater}\right]\middle| \signal;\strategy \right] -\E\left[\b(\stater) k \E\left[ \ractoth\vert \stater, \distq^{\stater}\right]\middle| \signal;\strategy \right] \\
& =\E\left[\b(\stater) k [ \cdist^{\stater} (\sigtup) - \cdistq^{\stater }(\sigtup)] \middle| \signal; \strategy\right]
\end{align*}
Since $\dist^\state \ccad \distq^\state$ for all $\state \in T$ and $\dist^\state = \distq^\state$ for all $\state \notin T$, the above expression is non-negative if $\signal \ge \sigt$ and non-positive otherwise. This implies that if $\strategy$ was an equilibrium under $\distq$, then it remains an equilibrium under $\dist$, because \eqref{Equation: IC to play a=1} continues to hold for all $\signal \ge \sigt$, and \eqref{Equation: IC to play a=0} continues to hold for all $\signal < \sigt$. Thus, $\strategy \in \eqc(\G,\dist)$, thus proving (1). \smallskip

Towards proving $(2)$, notice that $(a) \implies (b)$ was proved above, while $(b) \implies (c)$ is obvious. Therefore, we only need to prove that
$(c) \implies (a)$.

Suppose that $(\eqmaxp(\G,\dist), -\eqminp(\G,\dist)) \ge (\eqmaxp(\G,\distq), -\eqminp(\G,\distq))$ for all common-value affine coordination games $\G$ but, $\dist^{\state^*}\nccad \distq^{\state^*}$ for some $\state^* \in T$. Then there exists $\tilde \signal, \sighat$ such that at least one of the following holds:
\begin{enumerate}
    \item $\sighat \le \tilde \signal$ and $\dist^{\state^*}_{\tilde \signal}(\sighup) < \distq^{\state^*}_{\tilde \signal}(\sighup)$, or
    \item $\sighat \ge \tilde \signal$ and $\dist^{\state^*}_{\tilde \signal}(\sighdn) < \distq^{\state^*}_{\tilde \signal}(\sighdn)$.
\end{enumerate}
\smallskip
\noindent\textbf{Case 1: Suppose $\sighat \le \tilde \signal$ and $\dist^{\state^*}_{\tilde \signal}(\sighup) < \distq^{\state^*}_{\tilde \signal}(\sighup)$}

\noindent
We will construct a common-value affine coordination game with payoff functions $\payoff(\actoth,\state) = \a(\state) + \b(\state) h(\actoth)$, and a strategy profile $\strategy$ with \begin{equation}\label{strategyconstruction1}
P(\strategy)=\sighup \qquad \textrm{and} \qquad  NP(\strategy)=\sigset\setminus \sighup,\end{equation}
such that,
\begin{enumerate}[(i)]
    \item $\strategy \in \eqc(\G,\distq)$, but $\strategy \notin \eqc(\G,\dist)$.
    \item $\strategy$ is the equilibrium with maximal participation in $\distq$.
    \item No $\strategy'$ such that $\distq(P(\strategy')) > \distq(P(\strategy))$ is an equilibrium under $\dist$. Since the marginal distribution of $\dist$ and $\distq$ coincide, this would complete the proof in Case 1.
\end{enumerate}
\noindent
To this end, define \begin{align*}
    A:=& \{\psig: \signal \in \sigset, \signal < \sighat\} \\
    \text{ and } B:=& \{ \psig: \signal\in\sigset, \signal \ge \sighat, \signal \neq \tilde \signal\}.
\end{align*}
$A \bigcup \{\posterior(\tilde\signal)\} \bigcup B$ is an affinely independent set. By Lemma~\ref{Lemma: affine independence and separation}, there exists $\at \in \real^\nplayers$ such that $$\max_{\psig \in A} \at' \psig < \at' \posterior(\tilde\signal) < \min_{\psig \in B} \at' \psig.$$
Notice that $\at' \posterior(\cdot) = \E_{\posterior(\cdot)} [ \at(\stater)]$. Therefore, $\exists \at : \states \to \real$ such that,
\begin{align*}
    \max_{\psig \in A} \E_{\psig}[\at(\stater)] < \E_{\posterior(\tilde \signal)}[\at(\stater)] < \min_{\psig \in B} \E_{\psig}[\at(\stater)].
\end{align*}
\noindent
Define, for any $\signal$,
\[l(\distq, \signal):= \psig(\state^*) \cdistq^{\state^*}(\sighup).\]
Define, for any $k \in \real$,
\begin{align*}
     \De_1(k) &:=  \left[\E_{\posterior(\sigt)}[\at(\stater)] + kl(\distq,\sigt)\right]
     - \left[\max_{\psig \in A} \E_{\psig}[\at(\stater)] + k \right], \\
    \text{ and } \De_2(k) &:= \min_{\psig \in B} \left[\E_{\psig}[\at(\stater)] + kl(\distq,\signal)\right]
    - \left[\E_{\posterior(\sigt)}[\at(\stater)] + kl(\distq,\sigt)\right]
\end{align*}

Notice that $(\De_1(0), \De_2(0))> (0,0)$. By continuity in $k$, $\exists k > 0$ such that $(\De_1(k), \De_2(k)) > (0,0)$. Fix any such $k$, and choose a scalar $a$ so that $$\E_{\posterior(\sigt)} [ \at(\stater) + a] + k l(\distq,\sigt) = 0.$$ Define $$\a(\cdot) := \at(\cdot) + a,$$
and let $$\displaystyle \beta(\state) =\dfrac{ k }{(\nplayers-1)}\ind_{\state = \state^*}.$$

\noindent
Consider a common-value affine coordination game $\G$ with $h(x) = x$, $(\a,\b)$ defined above, and the strategy profile in (\ref{strategyconstruction1}) above.

Note that we have
\begin{align*}
    & \max_{\psig \in A} \left\{\E_{\psig}[\a(\stater)] + kl(\distq,\signal) \right\} \\ & < \max_{\psig \in A} \left\{\E_{\psig}[\a(\stater)] + k \right\}\\
    & < \E_{\posterior(\sigt)}[\a(\stater)] + k l(\distq,\sigt) \\
     = 0 &< \min_{\psig\in B} \left\{\E_{\psig}[\a(\stater)] + k l(\distq,\signal)\right\}
\end{align*}
\noindent
Further, for any $\signal$, $$\E_{\psig}[\b(\stater) (\nplayers-1) \cdistq^{\stater}(\sighup)] = k l(\distq,\signal).$$
Therefore the following conditions are satisfied.
\begin{align}
    \max_{\psig \in A} \left\{\E_{\psig}[\a(\stater)] + \E_{\psig}[\b(\stater) (\nplayers-1) \cdistq^{\stater}(\sighup)] \right\} <&   0 \notag \\
    \E_{\posterior(\sigt)}[\a(\stater)] + \E_{\posterior(\sigt)}[\b(\stater) (\nplayers-1)\distq^{\stater}_{\sigt}(\sighup)] =& 0 \label{Equation: cutoff equilibrium IC with beta depend on state}\\
    \min_{\psig\in B} \left\{\E_{\psig}[\a(\stater)] + \E_{\psig}[\b(\stater) (\nplayers-1)\cdistq^{\stater}(\sighup)]\right\} >& 0 \notag.
\end{align}
\noindent

Equations~\eqref{Equation: cutoff equilibrium IC with beta depend on state} simply imply that the strategy $P(\strategy)=\sighup$ and $NP(\strategy)=\sigset\setminus \sighup$ constitutes an equilibrium under $\distq$, i.e.,
$\strategy \in \eqmaxp(\G,\distq)$. In particular, \eqref{Equation: IC to play a=1} holds with an equality for $\tilde \signal$, and holds strictly for any $\signal > \sighat$ such that $\signal \neq \tilde\signal$.
Finally, notice that, since $\dist^{\state^*}_{\tilde \signal}(\sighup) < \distq^{\state^*}_{\tilde \signal}(\sighup)$,
\begin{align*}
    \E_{\posterior(\tilde\signal)}[\a(\stater)] + \E_{\posterior(\sigt)}[\b(\stater) \dist^{\stater}_{\tilde \signal}(\sighup)] &= \E_{\posterior(\tilde\signal)}[\a(\stater)] + k \posterior(\sigt)(\state^*)\dist_{\sigt}^{\state^*}(\sighup) \\
    &< \E_{\posterior(\tilde\signal)}[\a(\stater)] + k l(\distq,\sigt) =  0.
\end{align*}
Therefore, $\strategy \notin \eqc(\G,\dist)$.

\noindent Now, we will show that if $\strategy'$ is such that $P(\strategy') \supset P(\strategy)$, then $\strategy' \notin \eqc(\G,\dist) \bigcup \eqc(\G,\distq)$. To this end, for any such $\strategy'$, we have $P(\strategy') = \signal'^\uparrow$ for some $\signal' < \sighat$. For $\strategy'$ to be an equilibrium, \eqref{Equation: IC to play a=1} must be satisfied for all $\signal \ge \signal'$. Since, $\signal' < \sighat$, $\signal' \in A$. However, for all $\signal \in A$, we have that,
\begin{align*}
    &  \E_{\psig}[\a(\stater)] + k\max\{l(\distq,\signal), l(\dist,\signal)\}
    \\ & \le \E_{\psig}[\a(\stater)] + k \\
    & < \E_{\posterior(\sigt)}[\a(\stater)] + k l(\distq,\sigt) \\
     & = 0.
\end{align*}
Therefore, \eqref{Equation: IC to play a=1} cannot be satisfied for any relevant $\signal \in A \bigcap P(\strategy')$ under $\dist$ or $\distq$. Hence, if $\strategy' \in \eqc(\G,\dist)$, then $P(\strategy') \subset P(\strategy)$. Therefore, $\eqmaxp(\G,\dist) < \eqmaxp(\G,\distq)$, a contradiction.

\smallskip

\noindent \textbf{Case 2: $\sighat \ge \tilde \signal$ and $\dist^{\state^*}_{\tilde \signal}(\sighdn) < \distq^{\state^*}_{\tilde \signal}(\sighdn)$: }

\noindent The proof is nearly identical to the previous case. We will construct a strategy profile $\strategy$ such that,
\begin{enumerate}[(i)]
    \item $\strategy \in \eqc(\G,\distq)$, but $\strategy \notin \eqc(\G,\dist)$.
    \item $\strategy$ is the equilibrium with minimal participation in $\distq$.
    \item No $\strategy'$ such that $\distq(P(\strategy')) < \distq(P(\strategy)$ is an equilibrium under $\dist$. Since the marginal distribution of $\dist$ and $\distq$ coincide, this would complete the proof under Case 2.
\end{enumerate} Define
\begin{align*}
    A:=& \{ \psig: \signal \in \sigset, \signal \le \sighat, , \signal \neq \tilde\signal\}\\
    \text{ and } B: =& \{ \psig : \signal \in \sigset, \signal > \sighat\}.
\end{align*}

Notice that $\dist^{\state^*}_{\tilde \signal}(\sighdn) < \distq^{\state^*}_{\tilde \signal}(\sighdn)$ $\Longleftrightarrow $ $\dist^{\state^*}_{\tilde \signal}(B) > \distq^{\state^*}_{\tilde \signal}(B)$ since $B = \sigset \backslash \{ \sighdn\}$.
As before, because $A \bigcup \{\posterior(\tilde\signal)\} \bigcup B$ is an affinely independent set, Lemma~\ref{Lemma: affine independence and separation} implies that $\exists \at : \states \to \real$ such that,
\begin{align*}
    \max_{\psig \in A} \E_{\psig}[\at(\stater)] < \E_{\posterior(\tilde\signal)}[\at(\stater)] <   \min_{\psig \in B} \E_{\psig}[\at(\stater)].
\end{align*}

\noindent
Define for any $\signal$,
\[l(\distq,\signal):=\psig(\state^*) \cdistq^{\state^*}(B)\]

Define, for any $k \in \real$,
\begin{align*}
     \De_1(k) &:= \left( (\E_{\posterior(\sigt)}[\at(\stater)] + kl(\distq,\sigt)) - (\max_{\psig \in A} \E_{\psig}[\at(\stater)] +k l(\distq,\signal)   ) \right),  \\
    \text{ and } \De_2(k) &:= \left( \min_{\psig\in B} \E_{\psig}[\at(\stater)]   - (\E_{\posterior(\sigt)}[\at(\stater)] + kl(\distq,\sigt))\right)
\end{align*}
Notice that $(\De_1(0), \De_2(0))> (0,0)$. By continuity in $k$, $\exists k > 0$ such that $(\De_1(k), \De_2(k)) > (0,0)$. For any such $k$, choose a scalar $a$ so that $$\E_{\posterior(\sigt)} [ \at(\stater) + a] + k l(\distq,\sigt) = 0.$$

\noindent
Define $$\a(\cdot) := \at(\cdot) + a$$
and let
$$\displaystyle \beta(\state) =\dfrac{ k }{(\nplayers-1)}\ind_{\state = \state^*}.$$
Now define a common-value affine coordination game $\G$ with $h(x) = x$, and $(\a,\b)$ as above, and consider the strategy profile $\strategy = \ind_{B}$.
Note that we have
\begin{align*}
    \max_{\psig \in A} \left\{\E_{\psig}[\a(\stater)] +k l(\distq,\signal) \right\} < \E_{\posterior(\sigt)}[\a(\stater)] + k l(\distq,\sigt) = 0
     &< \min_{\psig\in B} \left\{\E_{\psig}[\a(\stater)]\right\}  \\
     &\le  \min_{\psig\in B} \left\{\E_{\psig}[\a(\stater)] + k l(\distq,\signal)\right\}
\end{align*}
Further note that for any $\signal$,
$$\E_{\psig}[\b(\stater) (\nplayers-1) \cdistq^{\stater}(B)] = kl(\distq,\signal).$$
Therefore the following conditions are satisfied.

\begin{align}
    \max_{\psig \in A} \left\{\E_{\psig}[\a(\stater)] + \E_{\psig}[\b(\stater) (\nplayers-1) \cdistq^{\stater}(B)] \right\} <&   0 \notag \\
    \E_{\posterior(\sigt)}[\a(\stater)] + \E_{\posterior(\sigt)}[\b(\stater) (\nplayers-1)\distq^{\stater}_{\sigt}(B)] =& 0 \label{Equation: cutoff equilibrium IC with beta depend on state, second}\\
    \min_{\psig\in B} \left\{\E_{\psig}[\a(\stater)] + \E_{\psig}[\b(\stater) (\nplayers-1)\cdistq^{\stater}(B)]\right\} >& 0 \notag.
\end{align}
Equations~\eqref{Equation: cutoff equilibrium IC with beta depend on state, second} imply that the strategy profile $\strategy = \ind_{B}$ constitutes an equilibrium under $\distq$, i.e., $\strategy \in \eqc(\G,\distq)$. Moreover, \eqref{Equation: IC to play a=0} holds with equality for $\tilde\signal$, and strictly for any $\signal \le  \sighat$ such that $\signal \neq \sigt$. Finally, notice that, since $\dist^{\state^*}_{\tilde \signal}(B) > \distq^{\state^*}_{\tilde \signal}(B)$,
\begin{align*}
    \E_{\posterior(\tilde\signal)}[\a(\stater)] + \E_{\posterior(\sigt)}[\b(\stater)(N-1) \distq^{\stater}_{\tilde \signal}(B)] &= \E_{\posterior(\tilde\signal)}[\a(\stater)] + k \posterior(\sigt)(\state^*)\dist_{\sigt}^{\state^*}(B) \\
    &> \E_{\posterior(\tilde\signal)}[\a(\stater)] + k l(\distq,\sigt) = 0.
\end{align*}
Therefore, $\strategy \notin \eqc(\G,\dist)$.

\noindent Now, we will show that if $P(\strategy') \subset B= P(\strategy)$, then $\strategy' \notin \eqc(\G,\dist) \bigcup \eqc(\G,\distq)$. Suppose not, i.e., $P(\strategy') \subset B$. Fix a $\signal' \in B \backslash P(\strategy')$.
For $\strategy'$ to be an equilibrium, \eqref{Equation: IC to play a=0} must be satisfied at $\signal'$. However, since $\signal' \in B$,
\begin{align*}
    &  \E_{\posterior(\signal')}[\a(\stater)] + k\min\{l(\distq,\signal'), l(\dist,\signal')\}\\
    &  \ge \E_{\posterior(\signal')}[\a(\stater)] \\
     & > \E_{\posterior(\sigt)}[\a(\stater)] + k l(\distq,\sigt) \\
     & = 0.
\end{align*}

Therefore, \eqref{Equation: IC to play a=0} cannot be satisfied at $\signal'$ under $\dist$ or $\distq$. Therefore, $\eqminp(\G,\distq) = \distq(P(\strategy))$. Finally, since $\strategy \notin \eqc(\G,\dist)$, $\eqminp(\G,\dist) > \eqminp(\G,\distq)$, a contradiction.
\eprf

\subsection{Contour-CAD and Separable Games}
\label{subsection: separable}

In this section, we consider a special case of the setting of Section~\ref{sec:contourcad}, where  $\beta(\cdot)$ is independent of the state. We call such games ``Separable common value affine coordination games.'' The net payoff of a player from taking action~$a=1$ in such a game is given by
\begin{align*}
    \payoff(\actoth,\state) = \a(\state) + \b h(\actoth).
\end{align*}
for some affine, increasing $h(\cdot)$ and $\b \ge 0$.
\noindent

An analog of Theorem~\ref{Theorem: CCAD with state} holds. The only difference is that we do not need cCAD increase for each state, but only in expectations.

\bthm\label{Theorem: CCAD equivalence} Let $\dist$ and $\distq$ be two distributions with identical marginal distributions.

\begin{enumerate}
    \item $\dist \ccad \distq \implies \eqc(\G,\dist) \supseteq \eqc(\G,\distq)$ for all separable common-value affine coordination games $\G$.

    \item Suppose $\distq$ (and hence $\dist$) is affinely independent. Then $\eqc(\G,\dist)\supseteq \eqc(\G,\distq)$ for all separable common-value affine coordination games $\G$ $\implies$ $\dist \ccad \distq$.
\end{enumerate}
\ethm

\bprf[Proof of Theorem~\ref{Theorem: CCAD equivalence}]The proof approach is similar. We first prove $(1)$. Suppose that $\dist \ccad \distq$. Consider a separable common-value affine  coordination game $\G$. Let $h(x) = k x + l$ for some $k \ge 0$. Suppose that $\strategy$ is a cutoff equilibrium under $\distq$ with a cutoff $\sigt$. Therefore,
\begin{align*}
    \E\left[\payoff(\ractoth, \stater) \vert \signal,\strategy;\distq\right] \ge 0 \quad \forall \signal \ge \tilde \signal\\
    \E\left[\payoff(\ractoth,\stater) \vert \signal,\strategy;\distq\right] \le 0 \quad \forall \signal < \tilde \signal
\end{align*}
For any $\signal$,
\begin{align*}
&\E\left[ \payoff(\ractoth, \stater) \vert \signal,\strategy;\dist \right] - \E\left[ \payoff(\ractoth, \stater) \vert \signal,\strategy;\distq \right]
\\
& =\E\left[ \a(\stater)+\b h(\ractoth)\vert \signal,\strategy;\dist \right] - \left[\a(\stater)+\b h(\ractoth) \vert \signal,\strategy;\distq \right]\\
& =\b k \left(\E\left[ \ractoth\vert \signal,\strategy;\dist \right] - \E\left[\ractoth \vert \signal,\strategy;\distq \right]\right)\\
& =\b k [ \cdist(\sigtup) - \cdistq(\sigtup)]
\end{align*}
Since $\dist \ccad \distq$, the above expression is non-negative if $\signal \ge \sigt$ and non-positive otherwise. This implies that if $\strategy$ was an equilibrium under $\distq$, then it remains an equilibrium under $\dist$, because \eqref{Equation: IC to play a=1} continues to hold for all $\signal \ge \sigt$, and \eqref{Equation: IC to play a=0} continues to hold for all $\signal < \sigt$. Thus, $\strategy \in \eqc(\G,\dist)$. \smallskip

Next we prove $(2)$. Suppose that $\eqc(\G,\dist) \supseteq \eqc(\G,\distq)$ for all separable common-value affine  coordination games $\G$ but $\dist \nccad \distq$. Then there exists $\tilde \signal, \sighat$ such that at least one of the following holds:
\begin{enumerate}
    \item $\sighat \le \tilde \signal$ and $\dist_{\tilde \signal}(\sighup) < \distq_{\tilde \signal}(\sighup)$, or
    \item $\sighat \ge \tilde \signal$ and $\dist_{\tilde \signal}(\sighdn) < \distq_{\tilde \signal}(\sighdn)$.
\end{enumerate}
\smallskip
\noindent\textbf{Case 1: Suppose $\sighat \le \tilde \signal$ and $\sighat \le \tilde \signal$ and $\dist_{\tilde \signal}(\sighup) < \distq_{\tilde \signal}(\sighup)$}

\noindent Define $$A:= \{\psig: \signal \in \sigset, \signal < \sighat\}.$$ $$B:= \{ \psig: \signal\in\sigset, \signal \ge \sighat, \signal \neq \tilde \signal\}.$$ Now $A \bigcup \{\posterior(\tilde\signal)\} \bigcup B$ is an affinely independent set. By Lemma~\ref{Lemma: affine independence and separation}, $\exists \lm \in \real^\nplayers$ such that $$\max_{\psig \in A} \lm' \psig < \lm' \posterior(\tilde\signal) < \min_{\psig \in B} \lm' \psig.$$
Notice that $\lm' \posterior(\cdot) = \E_{\posterior(\cdot)} [ \lm\stater]$. Therefore, $\exists \lm : \states \to \real$ such that,
\begin{align*}
    \max_{\psig \in A} \E_{\psig}[\lm(\stater)] < \E_{\posterior(\tilde \signal)}[\lm(\stater)] < \min_{\psig \in B} \E_{\psig}[\lm(\stater)].
\end{align*}

 Define $$\a(\cdot):= \lm(\cdot) \displaystyle -\frac{\E_{\posterior(\tilde \signal)}[\lm(\stater)] + \min_{\psig \in B}\E_{\psig}[\lm(\stater)]  }{2}.$$
Notice that $\E_{\posterior(\tilde\signal)}[\a(\stater)] = \displaystyle \frac{\E_{\posterior(\tilde \signal)}[\lm(\stater)] - \min_{\psig\in B}\E_{\psig}[\lm(\stater)]}{2} < 0$.
Therefore, we have
\begin{align}
    \max_{\psig \in A} \E_{\psig}[\a(\stater)] <  \E_{\posterior(\tilde \signal)}[\a(\stater)] < 0 < \min_{\psig \in B} \E_{\psig}[\a(\stater)]. \label{Equation: separation alpha}
\end{align}
Define $$\displaystyle \b:= -\frac{\E_{\posterior(\tilde \signal)}[\a(\stater)]}{ \distq_{\tilde \signal}(\sighup)}.$$
Notice that $\b > 0$. Also, \eqref{Equation: separation alpha} and the monotonicity of
$\distq_{\signal}(\sighup)$ (in $\signal$)

imply,
\begin{align}
    \max_{\psig \in A} \left\{\E_{\psig}[\a(\stater)] + \b \cdistq(\sighup) \right\} <&   0 \notag \\
    \E_{\posterior(\tilde\signal)}[\a(\stater)] + \b \distq_{\tilde \signal}(\sighup) =& 0 \label{Equation: cutoff equilibrium IC}\\
    \min_{\psig\in B} \left\{\E_{\psig}[\a(\stater)] + \b \cdistq(\sighup)\right\} >& 0 \notag.
\end{align}
Notice that it is possible that $\cdistq(\sighup) < \distq_{\tilde \signal}(\sighup)$. However, the last inequality follows from $\E_{\posterior(\tilde \signal)}[\a(\stater)] < 0 < \min_{\psig} \E_{\psig}[\a(\stater)]$ and $\cdistq(\sighup) \ge 0$.

Consider a separable common-value affine  coordination game $\G$ with $h(x) = x$, and $(\a,\b)$ defined above. Consider a strategy profile $\strategy$ such that $P(\strategy)=\sighup$ and $NP(\strategy)=\sigset\setminus \sighup$. Then, Equations~\eqref{Equation: cutoff equilibrium IC} simply says that this strategy constitutes an equilibrium under $\distq$, i.e.,
$\strategy \in \eqc(\G,\distq)$. Moreover, \eqref{Equation: IC to play a=1} holds with an equality for $\tilde \signal$, and holds strictly for any $\signal \ge \sighat$ such that $\signal \neq \tilde\signal$. Finally, notice that $$\dist_{\tilde \signal}(\sighup) < \distq_{\tilde \signal}(\sighup) \implies \E_{\posterior(\tilde\signal)}[\a(\stater)] + \b \dist_{\tilde \signal}(\sighup) < 0.$$ Therefore, $\strategy \notin \eqc(\G,\dist)$, a contradiction. Thus, case 1 cannot hold.

\smallskip

\noindent \textbf{Case 2: $\sighat \ge \tilde \signal$ and $\dist_{\tilde \signal}(\sighdn) < \distq_{\tilde \signal}(\sighdn)$: }

\noindent Define $A:= \{ \psig: \signal \in \sigset, \signal \le \sighat, , \signal \neq \tilde\signal\}$ and $B: = \{ \psig : \signal \in \sigset, \signal > \sighat\}$.
Notice that $\dist_{\tilde \signal}(\sighdn) < \distq_{\tilde \signal}(\sighdn)$ $\Longleftrightarrow $ $\dist_{\tilde \signal}(B) > \distq_{\tilde \signal}(B)$ since $B = \sigset \backslash \{ \sighdn\}$.

Proceeding as before, we use the fact that $A \bigcup \{\posterior(\tilde\signal)\} \bigcup B$ is an affinely independent set. Therefore, by Lemma~\ref{Lemma: affine independence and separation}, $\exists \lm : \states \to \real$ such that,
\begin{align*}
    \max_{\psig \in A} \E_{\psig}[\lm(\stater)] < \E_{\posterior(\tilde\signal)}[\lm(\stater)] < 0 <  \min_{\psig \in B} \E_{\psig}[\lm(\stater)].
\end{align*}
Let $a > 0 $ be such that, for all $\signal \le \sighat, \signal \neq \tilde \signal$,
\begin{align*} a \left[ \E_{\posterior(\tilde\signal)}[ \lm(\stater) - \E_{\psig}[\lm(\stater)]\right] > \max_{\signal \le \sighat} ( \cdistq(B) - \distq_{\tilde \signal}(B)).\end{align*}
Such an $a$ exists since $\displaystyle \E_{\posterior(\tilde \signal)}[\lm(\stater)]>  \max_{\psig \in A} \E_{\psig}[\lm(\stater)].$ Using this $a$, define $\a(\cdot) = a \lm(\cdot) + b$ where $\displaystyle b:= -(a \E_{\posterior(\tilde\signal)}[\lm(\stater)] + \distq_{\tilde \signal}(B)).$ By definition, $\E_{\posterior(\tilde \signal)}[\a(\stater)] + \distq_{\tilde \signal}(B) = 0$. By the choice of $a$, for all $\signal \in A$,
\begin{align*}
0 = \E_{\posterior(\tilde\signal)}[\a(\stater)] + \distq_{\tilde \signal}(B) > \E_{\psig}[\a(\stater)] + \cdistq(B).
\end{align*}
Finally, by the monotonicity of $\cdistq(B)$ (in $\signal$), for all $\signal \in B$,
\begin{align*}
    0 = \E_{\posterior(\tilde\signal)}[\a(\stater)] + \distq_{\tilde \signal}(B) < \E_{\psig}[\a(\stater)] + \cdistq(B)
\end{align*}
Therefore, by letting $\b :=1$, we have,

\begin{align}
    \max_{\psig \in A} \left\{\E_{\psig}[\a(\stater)] + \b \cdistq(B) \right\} <&   0 \notag \\
    \E_{\posterior(\tilde\signal)}[\a(\stater)] + \b \distq_{\tilde \signal}(B) =& 0 \label{Equation: cutoff equilibrium IC, second}\\
    \min_{\psig\in B} \left\{\E_{\psig}[\a(\stater)] + \b \cdistq(B)\right\} >& 0 \notag.
\end{align}

Similar to before, consider a game $\G$ with $h(x) = x$, and $(\a,\b)$ as obtained above. Notice that the equations \eqref{Equation: cutoff equilibrium IC, second} are simply \eqref{Equation: IC to play a=1} and \eqref{Equation: IC to play a=0} for the strategy profile $\strategy = \ind_{B}$. Thus, $\strategy \in \eqc(\G,\distq)$. Moreover, \eqref{Equation: IC to play a=0} holds with an equality for $\tilde\signal$, and holds strictly for any $\signal \le  \sighat$ such that $\signal \neq \tilde\signal$. Finally, notice that $$\dist_{\tilde \signal}(B) > \distq_{\tilde \signal}(B) \implies \E_{\posterior(\tilde\signal)}[\a(\stater)] + \b \dist_{\tilde \signal}(B) > 0.$$ Therefore, $\strategy \notin \eqc(\G,\dist)$, a contradiction. Thus, case 2 cannot hold either.
This completes the proof of Theorem~\ref{Theorem: CCAD equivalence}.
\eprf

\subsection{Proof of Theorem~\ref{Theorem: sCAD equivalence}}

\bprf
The proof approach is identical to that of Theorem~\ref{Theorem: wCAD equivalence}.
First, we establish $(1) \implies (2)$. Since $\dist \scad \distq$ for all increasing $h(\cdot)$, $$\E[ h(\ractoth) \vert  \signal;\dist] \ge \E[h(\ractoth)\vert \signal;\distq] \quad \text{ for all } \signal \in \sigset.$$
Given any $\strategy$, we can replicating the same argument, to obtain:

\bigskip

\noindent
If $\b(\signal) \ge 0 \quad \forall  \signal \in \sigset$,
\[
\begin{array}{ll}
\E[ \payoff(\ractoth, \signal) \vert \signal,\strategy;\dist ] - \E[ \payoff(\ractoth, \signal) \vert \signal,\strategy;\distq ] &\ge 0   \quad \text{ if } \signal \in P(\strategy) \\
&\le 0   \quad \text{ if } \signal \in NP(\strategy)
\end{array}
\]
Therefore, for all $\G$, if $\b(\signal) \ge 0$ for all $\signal\in\sigset$, then $\strategy \in \eq(\G,\distq) \implies \strategy \in \eq(\G,\dist)$.

Next, we establish $(2) \implies (1)$. We only sketch the argument for $(2) \implies (1)$, as the logic is very similar to that in Theorem~\ref{Theorem: wCAD equivalence}. Suppose that $\eq(\G, \dist) \supseteq \eq(\G,\distq)$ for all $\G$ with $\b(\signal) \ge 0 $ for all $\signal \in \sigset$, but $\dist \nscad \distq$. Therefore, $\exists \signal^*\in \sigset$ and $K \ni \signal^*$ such that, for some $m < \nplayers -1$,
$$\dist(\{\countr(K) \ge m\}  \vert \sigr_i = \signal^*) < \distq(\{\countr(K) \ge m\} \vert \sigry_i = \signal^*).$$

We will construct a game $\G$ that exhibits strategic complementarity such that $\strategy \in \eq(\G,\distq)$ but $\strategy \notin \eq(\G,\dist)$. Using the same idea as before, let $\strategy$ be a strategy such that
\[ P(\strategy)=K \qquad NP(\strategy)=\sigset \setminus K.\]

We specify a game $\G$ for which $\strategy$ constitutes an equilibrium under $\distq$.
\begin{align*}
 h(\actoth) &= \ind_{\actoth \ge m}& \\
 \a(\state) &= -1 & \\
 \b(\state,\signal) &\begin{cases}
 = 0 &    \text{ if }  \signal \notin K\\
 = \frac{\displaystyle 1}{\displaystyle \E\left[ \ind_{\ractoth \ge m}  \bigg \vert \signal^*;\distq \right]} & \text{ if } \signal \in K , \signal = \signal^*\\
    \ge \frac{\displaystyle 1}{\displaystyle  \E\left[\ind_{\ractoth \ge m} \bigg\vert \signal; \distq \right]} & \text{ if } \signal \in K, \signal \neq \signal^*
\end{cases}
\end{align*}
It is straightforward to verify that $\strategy \in \eq(\G,\distq) \backslash \eq(\G,\dist)$, a contradiction.
\eprf

\subsection{Proof of Proposition~\ref{proposition:congestion}}

\bprf

Using a similar reasoning as in Theorem~\ref{Theorem: wCAD equivalence}, we can see that if $\b(\signal) \le 0$ for all $\signal \in \sigset$, then for strategy profile $\strategy$, the inequalities in Equation~(\ref{complementarityproof}) would be reversed, i.e.,
\newline For $\b(\signal) \le 0$,
\begin{align*}
\E[ \payoff(\ractoth, \signal) \vert \signal,\strategy;\dist ] - \E[ \payoff(\ractoth, \signal) \vert \signal,\strategy;\distq ] &\le  0 \quad \text{ if } \signal \in P(\strategy)\\
&\ge  0 \quad \text{ if } \signal \in NP(\strategy)
\end{align*}
Therefore, if $\strategy \in \eq(\dist,\G)$, then \eqref{Equation: IC to play a=1} continues to hold for all $\signal \in P(\strategy)$ under $\distq$ and \eqref{Equation: IC to play a=0} holds for all $\signal \in NP(\strategy)$ under $\distq$. Therefore, $\strategy \in \eq(\G,\distq)$.

Suppose that $\eq(\G, \distq) \supseteq \eq(\G,\dist)$ for all games $\G$ with $\b(\signal) \le 0 $ for all $\signal \in \sigset$, but $\dist \ncad \distq$. Using Lemma~\ref{Lemma: wCAD equivalent with increasing expectation} again, we know that $\exists \signal^*$ and $K \ni \signal^*$ such that,
$$\dist(\sigrs_j \in  K \vert \sigrs_i = \signal^*) < \distq (\sigrs_j \in K \vert \sigrs_i = \signal^*).$$
\noindent
Now, we define a $\strategy$ with $P(\strategy)=K$, and
a game $\G$ as follows:
\[
 \begin{array}{ll}
  \a(\signal) &= \begin{cases} 1 &  \text{if} \ \signal \in K \\ -1 & \text{if} \ s\notin K\end{cases} \\
  h(\actoth) &= \actoth\\
  \b(\signal) &  \left\{
 \begin{array}{ll}
  = 0    & \text{ if }  \signal \notin K\\
  =\frac{\displaystyle -1}{\displaystyle  \E\left[ \ractoth \bigg \vert \signal^*;\dist\right]} & \text{ if } \signal \in K \text{ and } \signal = \signal^*\\
   \ge \frac{\displaystyle - 1}{\displaystyle \E\left[\ractoth \bigg\vert \signal; \dist \right]} & \text{ if } \signal \in K \text{ and } \signal \neq \signal^*.
  \end{array}
 \right.
 \end{array}
 \]
 It is straightforward to verify that $\strategy \in \eq(\G,\dist)$, but $\strategy \notin \eq(\G,\distq)$, a contradiction.

This proves the first part of the proposition (CAD equivalence). The second part (sCAD equivalence) makes a the same modification to the reasoning in Theorem~\ref{Theorem: sCAD equivalence}, and hence omitted.
\eprf

\subsection{Proof of Proposition~\ref{Proposition: second price auction revenue}}
\bprf
Following \cite{meyer1990interdependence}, we define an elementary transformation on identical intervals (ETI) as follows: Given any $\signal < \signal'$, we increase the probabilities of $(\signal,\signal)$ and $(\signal',\signal')$ by some $\a > 0$, while the probabilities of $(\signal',\signal)$ and $(\signal,\signal')$ are decreased by $\a$ each. Proposition~1 of \cite{meyer1990interdependence} says that $\dist \cad \distq$ if and only if $\dist$ can be obtained by a finite sequence of ETIs starting from $\distq$. We now establish that any ETI increases the expected revenue. To this end, fix $\signal < \signal'$ and consider an ETI involving $\signal$ and $\signal'$ to construct $\widehat \distq$ from $\distq$.

Notice that the price of the object is $\signal$ whenever the valuations are $(\signal,\signal), (\signal,\signal')$ or $(\signal',\signal)$. Therefore,
\begin{align*}
    & R(\widehat \distq) - R(\distq) = \a (\signal' + \signal) - 2 \a (\signal) = \a ( \signal' - \signal) > 0.
\end{align*}
Hence, $R(\dist) \ge R(\distq)$ with the inequality being strict whenever $\dist \neq \distq$.
\eprf

\newpage
\section{Online Appendix}

\subsection{Non-exchangeable signal distributions}
\label{subsec:nonexchangeability}
\subsection*{\small{The Game $\Gamma$ and Information Structure~$\dist$}}
There are $\nplayers$ players, indexed by $i$ or $j\neq i$. They simultaneously and independently choose whether to act ($a_i=1$) or not ($a_i=0$). Each player $i$'s payoff depends on the \emph{weighted} aggregate action by others. That is, $\exists \lmv \in \real^\nplayers_+$ such that $$\ractoth=\sum_{j\neq i} \lm_j a_j.$$ Rest of the setup remains the same as in the main model. Let $\cdist^i(\cdot) \in \De(\sigset^{\nplayers -1})$ denote the conditional distribution of player $i$ over other players' types, given that a player $i$ has type $\signal$. Unlike the main text, $\dist$ need not be exchangeable, this conditional distribution is indexed with the player identity. For any pair of agents~$i$ and $j$, and any $T \subset \sigset$, we define $$\cdist^i(\sigr_j \in T):=Prob(\sigr_j\in T \vert \sigr_i = \signal).$$

Suppose that $\vec s=(s_1,s_2,\ldots s_\nplayers)\in \sigset^\nplayers$ is the realized type profile. A player $i$ with type $\sigrs_i=s$ get a payoff of $u(a_i=1,\ractoth,\signal)$ if she acts and $u(a_i=0,\ractoth,\signal)$ if she does not act. The net payoff from taking action
$$\payoff(\ractoth,\signal):=u(a_i=1,\ractoth,\signal)-u(a_i=0,\ractoth,\signal)=\a(\signal)+\b(\signal)h(\ractoth),$$
where $h(\cdot)$ is increasing.

\begin{definition}
    The game $\G$ exhibits \textbf{strategic complementarity} if $\beta(\cdot)\ge 0$ and \textbf{strategic substitutability} if $\beta(\cdot)\le 0$ for all $s\in \sigset$. The game is said to be \textbf{affine} if $h(\cdot)$ is affine.
\end{definition}

\bdefn[\textbf{Concentration along a Diagonal}]\label{Definition: Weak CAD-nonexch} We say that $\dist$ has a ``higher concentration along a diagonal'' than $\distq$, or $\dist$ \textbf{is CAD higher than} $\distq$, denoted by $\dist \cad \distq$, if,
\begin{enumerate}
    \item $\dist$ and $\distq$ have the same marginal distributions. And,
    \item For all $i,j\in\players, \signal \in \sigset$, and any pair of agents~$i,j$ \begin{enumerate}[(a)]
        \item $\cdist^i(\sigr_j = \signal) \ge \cdistq^i(\sigr_j = \signal)$, and
        \item $\cdist^i(\sigr_j = \signal') \le \cdistq^i(\sigr_j = \signal')$ whenever $\signal' \neq \signal$.
    \end{enumerate}
\end{enumerate}
\edefn
\noindent
If $\sigr \sim \dist$ and $\sigry \sim \distq$ with $\dist \cad \distq$, we say $\sigr \cad \sigry$, i.e., we use $\sigr \cad \sigry$ and $\dist\cad \distq$ interchangeably.

\blemma\label{Lemma: wCAD equivalent with increasing weighted expectation} Let $\sigr$ and $\sigry$ be two $\sigset^\nplayers$-valued random variables with distributions $\dist$ and $\distq$ respectively. Moreover, $\dist$ and $\distq$ have identical marginals. Then, the following are equivalent.
\begin{enumerate}
    \item $\dist \cad \distq$.
    \item For all $\signal \in \sigset$ and $K \subseteq \sigset$ such that $\signal \in K$, and $\lmv \in \real^\nplayers_+$,  $$\E\left[ \sum_{j \neq i} \lm_j \ind_{\sigr_j \in K} \bigg\vert \sigr_i = \signal \right] \ge \E\left[\sum_{j\neq i} \lm_j \ind_{\sigry_j \in K} \bigg\vert \sigry_i = \signal\right].$$
\end{enumerate}
\elemma
\bprf[Proof of Lemma~\ref{Lemma: wCAD equivalent with increasing weighted expectation}]
Fix $i$ such that $\sigr_i = \signal$. Then,
\begin{align*}
    \dist \cad \distq &\Longleftrightarrow \cdist^i(\sigr_j \in K \bigg\vert \sigr_i = \signal) \ge \cdist^i(\sigry_j \in K \bigg\vert \sigry_i = \signal) \quad \forall j, \signal, K \ni \signal, \\
    & \Longleftrightarrow \E\left[ \lm_j \ind_{\sigr_j \in K} \bigg\vert \sigr_i = \signal\right] \ge \E\left[ \lm_j \ind_{\sigry_j \in K} \bigg\vert \sigry_i = \signal\right] \quad \text{for any } \lm_j\in \real_+, \forall \signal, K \ni \signal. \\
    & \implies  \E\left[ \sum_{j\neq i} \lm_j \ind_{\sigr_j \in K} \bigg\vert \sigr_i = \signal\right] \ge  \E\left[\sum_{j\neq i}  \lm_j \ind_{\sigry_j \in K} \bigg\vert \sigry_i = \signal\right] \text{for any } \lm_j\in \real_+, \forall  \signal, K \ni \signal
    .
\end{align*}
For the reverse, suppose that $\dist \ncad \distq$. Therefore, $\exists i,j, \signal, K \ni \signal$, such that
\begin{align*}
    \cdist^i(\sigr_j \in K) < \cdistq^i(\sigr_j \in K).
\end{align*}
Then, let $\lmv$ be defined by $\lm_k = \ind_{k \in \{i,j\}}$. Then,
\begin{align*}
    \E\left[ \sum_{j\neq i} \lm_j \ind_{\sigr_j \in K} \bigg\vert \sigr_i = \signal\right] <  \E\left[\sum_{j\neq i}  \lm_j \ind_{\sigry_j \in K} \bigg\vert \sigry_i = \signal\right]
\end{align*}contradicting the hypothesis.
\eprf

\bthm\label{Theorem: wCAD equivalence without exch} Let $\dist$ and $\distq$ be two joint distributions over $\sigset^N$ and $h(\cdot)$ is affine and increasing. The following are equivalent.

\begin{enumerate}
    \item $\dist \cad \distq$.
    \item $\eq(\G, \dist) \supseteq \eq(\G,\distq)$ for all $\G$ that exhibits strategic complementarity.
    \item $\eq(\G, \dist) \subseteq \eq(\G,\distq)$ for all $\G$ that exhibits strategic substitutability.
\end{enumerate}
\ethm

\bprf[Proof of Theorem~\ref{Theorem: wCAD equivalence without exch}]
The proof of $(1) \implies (2), (3)$ is identical to the proof of Theorem~\ref{Theorem: wCAD equivalence} with the only difference being that the conditional distributions now have to be indexed by the player identity $i$. Therefore, we omit it.

\noindent
Now, we show $(2) \implies (1)$ and $(3) \implies (1)$. Suppose that $\eq(\G, \dist) \supseteq \eq(\G,\distq)$ for all $\G$ with $\b(\signal) \ge 0 $ for all $\signal \in \sigset$, but $\dist \ncad \distq$. Therefore, $\exists i,j$ a $\signal \in \sigset$, and $K \subseteq \sigset$ with $\signal \in K$, such that, $\cdist^i(\sigr_j \in K) < \cdistq^i(\sigr_j \in K).$ Let $\lmv$ be defined by $\lm_k = \ind_{k \in \{i,j\}}$. Then, this game is essentially a $2-$player symmetric game. Therefore, replicating the construction from the proof of Theorem~\ref{Theorem: wCAD equivalence} with $\nplayers = 2$ and using Lemma~\ref{Lemma: wCAD equivalent with increasing weighted expectation}, we obtain the desired equivalence.
\eprf

\subsection{Relationship of CAD with other orders}
\label{subsection: relation to other orders}
Below, we discuss how our proposed orders---strong CAD, CAD, and contour CAD---relate to well-known orders like the supermodular order or the concordance order. In two dimensions, the supermodular and the concordance orders are equivalent \cite[see][]{meyer2012increasing}. For more than two dimensions, the supermodular order is strictly stronger than the concordance order, i.e., $\dist \sm \distq \implies \dist \conc \distq$, where $\sm$ and $\conc$ denote the supermodular and the concordance orders respectively.

\bprop\label{Proposition: relationship with sm and conc}
Suppose $\dist$ and $\distq$ are joint distributions over $\sigset^\nplayers$-valued, exchangeable random variables.
\begin{enumerate}
    \item For $N=2$, $\dist \ccad \distq$ $\implies$ $\dist \sm \distq$. But $\exists$ $\dist, \distq$ such that $\dist \sm \distq$ but $\dist \nccad \distq$.
    \item For $N>2$, $\cad$ and $\sm$ are not nested. Also, $\ccad$ and $\sm$ are not nested.
\end{enumerate}
\eprop

\bprf[Proof of Proposition~\ref{Proposition: relationship with sm and conc}]
Proposition~\ref{Proposition: relation between orders} established that $\dist \cad \distq$ $\implies \dist \ccad \distq$. We first prove part $(1)$.
Suppose that $\dist \ccad \distq$. Consider any $\signal,\signal' \in \sigset$. Let $\signal' \ge \signal$ wlog. Then,
\begin{align*}
    \dist(\sigr_1 \le \signal, \sigr_2 \le \signal') & = \sum_{\hat \signal \le \signal} \dist_{\hat \signal}(\sigr_2 \le \signal') \dist(\sigr_1 = \hat \signal) \\
    & \ge  \sum_{\hat \signal \le \signal} \distq_{\hat \signal}(\sigr_2 \le \signal') \distq(\sigr_1 = \hat \signal)\\
    & =\distq(\sigr_1 \le \signal, \sigr_2 \le \signal')
\end{align*}
where the inequality is due to the fact that $\dist \ccad \distq$ and the marginal distributions of $\dist$ and $\distq$ coincide. Similarly,
\begin{align*}
    \dist(\sigr_1 \ge \signal, \sigr_2 \ge \signal') & = \sum_{\hat \signal \ge \signal'} \dist_{\hat \signal}(\sigr_1 \ge \signal) \dist(\sigr_2 = \hat \signal) \\
    & \ge  \sum_{\hat \signal \ge \signal'} \distq_{\hat \signal}(\sigr_1 \ge \signal) \distq(\sigr_2 = \hat \signal)\\
    & =\distq(\sigr_1 \ge \signal, \sigr_2 \ge \signal')
\end{align*} where, again, the inequality is due to the fact that $\dist \ccad \distq$ and the marginal distributions of $\dist$ and $\distq$ coincide. Therefore, $\dist \sm \distq$, and hence, $\dist \conc \distq$.
\medskip

Next, we prove part (2) by providing an example.\footnote{This example is due to Margaret Meyer and Bruno Strulovici, and can be found \href{https://people.maths.bris.ac.uk/~mb13434/prst_talks/M_Meyer_150522_PrSt_Bristol.pdf}{here}.}
Let $\sigset=\{0,1\}$ and $N=3$. Consider $\sigset^3$-valued random variable~$\vec \sigr:=(\sigr_1, \sigr_2, \sigr_3)$. Consider  two different joint distributions over $\vec \sigr$ denoted by~$\dist$ and $\distq$ as described in Table~\ref{Table: supermodular and weak-CAD} below.
\begin{table}[h!]
\centering
\begin{tabular}{|c|c|c|c|c|}
\hline
$\signal_1$ & $\signal_2$ & $\signal_3$ & $\dist$     & $\distq$     \\ \hline
0     & 0     & 0     & \rule{0pt}{2.5ex} $\frac{1}{3}$ & 0       \\ \hline
0     & 0     & 1     & \rule{0pt}{2.5ex} 0       & $\frac{1}{4}$ \\ \hline
0     & 1     & 0     & \rule{0pt}{2.5ex} 0       & $\frac{1}{4}$ \\ \hline
1     & 0     & 0     & \rule{0pt}{2.5ex} 0       & $\frac{1}{4}$ \\ \hline
0     & 1     & 1     & \rule{0pt}{2.5ex} $\frac{1}{6}$ & 0       \\ \hline
1     & 0     & 1     & \rule{0pt}{2.5ex} $\frac{1}{6}$ & 0       \\ \hline
1     & 1     & 0     & \rule{0pt}{2.5ex} $\frac{1}{6}$ & 0       \\ \hline
1     & 1     & 1     & \rule{0pt}{2.5ex} $\frac{1}{6}$ & $\frac{1}{4}$       \\ \hline
\end{tabular}
\caption{$\dist \cad \distq$ but $\dist$ and $\distq$ not $\sm$-ranked}
\label{Table: supermodular and weak-CAD}
\end{table}

Let us first show that $\dist \cad \distq$ but $\dist \notsm \distq$. Notice that, for any $i,j$ such that $i\neq j$ and $a \in \{0,1\}$, $\dist(\sigrs_i  = a \vert \sigrs_j = a) = \frac13 > \frac14 = \distq(\sigrs_i = a \vert \sigrs_j =a)$. Therefore, $\dist \cad \distq$. Consider the supermodular function $f(s) = \ind_{s = (1,1,1)}$. Then, $\E_\dist[f(\vec \sigrs)]] < \E_\distq[f(\vec \sigrs)]$. Also, consider $g(x) = \ind_{x = (0,0,0)}$. Then, $\E_\dist[g(\vec \sigrs)] > \E_\distq[g(\vec \sigrs)]$. Therefore, $\dist \notsm \distq$.

To show that the orders are not nested, however, we also need to show that there are exchangeable distributions $\dist,\distq$ such that $\dist \sm \distq$ but $\dist \ncad \distq$ for more than $2$ dimensions. The example below, in the spirit of the one in Figure~\ref{fig:example}, establishes this. Let $\vec \sigr$ be a $\{1,2,3\}^3-$valued random variable. Suppose that $\distq$ is the joint distribution of $\vec \sigr$ with $\sigr_i$'s being independent and uniformly distributed. Now, consider the following operations for some small $\a > 0$ (so that probabilities stay non-negative):
\begin{itemize}
    \item If the realization is $(2,2,2)$, reduce the mass by $6\a$.
    \item For any realization that is a permutation of $(2,1,2)$, increase the mass by $2 \a$.
    \item For any realization that is a permutation of $(2,2,3)$, increase the mass by $2 \a$.
    \item For any realization that is a permutation of $(1,2,3)$, reduce the mass by $\a$.
\end{itemize}
It is easy to check that $\dist(\sigr_2 = 2 \vert \sigr_1 = 2) < \distq(\sigr_2 = 2 \vert \sigr_1 = 2)$. Therefore, $\dist \ncad \distq$. In fact, $\dist \nccad \distq$ too since $$\dist(\sigr_2 \in \{1,2\} \vert \sigr_1 = 2) < \distq(\sigr_2 \in \{1,2\} \vert \sigr_1 = 2).$$ Finally, to check that $\dist \sm \distq$ necessitates proving that $\dist$ can be obtained from $\distq$ through ``elementary transformations'' as defined in \citet{meyer2015beyond}. This is indeed the case, however, we omit the details here.

\eprf

\subsection{Common learning and similarity of information}\label{Section: awaya krishna example}

\cite{awaya2022common} show how more correlation in players' signals impedes common learning thereby hindering coordination. We briefly discuss their leading example here to illustrate the key contrast.

There is a binary state of the world $\state \in \{G,B\}$. There are two players, and two periods. In each period, each agent receives some $\{0,1\}$-valued signal about the state. Signal realization of $1$ is conclusive of state $G$ while $0$ can arise in both states. Signals are independent across periods but can be correlated across players within a given period. At the end of period $2$, they decide simultaneously  and privately whether to invest at a cost $c$ or not. When $\state = B$, investment yields no returns, while in state $G$, an investment yields a return of $1$ if the other player invests. Thus, both players would like to coordinate to invest in state $G$ and not invest in state $B$.

They point out that when the signals are conditionally independent across players, there are cost parameters and signal structures such that investing upon receiving at least one $1$ in two periods is an equilibrium. Subsequently, for some such cost and signal structures, they do the analogous exercise as ours: they keep the marginal distribution unchanged and make the signals more similar (in the CAD or PQD order) within a period. Yet, the unique equilibrium is that no player invests regardless of the number of $1$s they receive!

The source of this finding is that increasing similarity \emph{within a period} does not imply that the vector of signals across the two periods becomes more similar. In particular, it is possible that when a player receives exactly one $1$ and one $0$, he assigns a \emph{higher} probability to the other agent receiving both zeroes with signals being more similar within a period. That is, we can view the two players' signals across two periods as a $\{(0,0),(0,1), (1,0), (1,1)\}^2$-valued random variable. Then, $Prob(\sigr_2 \in \{(0,1),(1,0),(1,1)\}\vert \sigr_1 = (0,1))$ may decrease with signals being more similar within a period. Here, $\sigr_i$ is player $i$'s signal in two periods.

\subsection{Information similarity and rationalizability}\label{rationalizability}
In this section, we explore the effect of increases in similarity in the CAD orders on the set of rationalizable actions  (rather than the set of equilibria).  Below, we consider the a private-value example, and show that any action that is rationalizable for a type continues to be rationalizable if information becomes more similar in the sense of CAD.

Let  $\sigset$ be the set of types and $\distq$ be the joint distribution of types. A mapping
$x:\sigset \to \mathbbm R$
describes the payoff parameter of interest to a player. Player~$i$ knows her own payoff parameter~$x(\signal_i)$ but may not know the other player's payoff parameter~$x(\signal_j)$.

$$\distq_\signal(T):=Prob(\sigrs_j\in T \vert \sigrs_i = \signal)$$
describes player $i$ of type $\signal$'s belief about the other player's type.

Each player decide whether or not to invest, and payoffs are as follows.
\medskip
\begin{center}
\begin{tabular}{|c|c|c|}
\hline
& Invest & Don't \\
\hline
Invest & $x(\signal_1), x(\signal_2)$ & $x(\signal_1)-1, 0$ \\
\hline
Don't & $0, x(\signal_2)-1$ & $0, 0$\\
\hline
\end{tabular}
\end{center}\medskip

\begin{proposition}
    \label{proposition:commonbelief}
    If $\dist\cad \distq$, then any action $a_i$ that is rationalizable for a type $\signal_i$ of player $i$ under $\distq$ remains rationalizable under $\dist$.
\end{proposition}

\begin{proof}

Consider $\phi_i:\sigset\to \mathbbm R$. Given the information structure $\distq$, an event $E=E_1\times E_2$ for $E_i\subseteq \sigset$ is said to be $(\phi_1,\phi_2)-$ believed if  each player $i$ of type $\signal_i$ believes that the probability of the event $E$ is at least $\phi_i(\signal_i)$. Formally, define
    $$B_i^{\phi_i}(E|\distq):=\{(\signal_1,\signal_2)|\signal_i\in E_i, \distq_{s_i}(E_j)
    \geq \phi_i(\signal_i)\}.$$

    Then, the set of states where the event $E$ is $(\phi_1,\phi_2)-$believed is
    $$B_*^{\phi_1,\phi_2}(E|\distq)=B_1^{\phi_1}(E|\distq)\bigcap B_2^{\phi_2}(E|\distq)$$
    There is a common $(\phi_1,\phi_2)-$ belief of the event $E$ if it is $(\phi_1,\phi_2)-$ believed that it is $(\phi_1,\phi_2)-$ believed, and so on. The set of states where the event $E$ is common $(\phi_1,\phi_2)-$ believed is
    $$C^{\phi_1,\phi_2}(E|\distq):=\bigcap_{n\geq 1} [B_*^{\phi_1,\phi_2}]^n(E|\distq)$$

It follows from Proposition 1 in \cite{morris2016common} that under the information structure $\distq$, for any type $\signal_i$ of player $i$, the action invest is rationalizable if $\signal_i\in C_i^{1-x,1-x}(T|\distq)$ and the action not invest is rationalizable if $\signal_i\in C_i^{x,x}(T|\distq)$. The argument is as follows.

Let $R_i^1(\distq)$ be the set of types such that investment is level 1 rationalizable for player $i$. If $x(\signal_i)< 0$, player $i$ gets a strictly higher payoff from not investing regardless of what the opponent does. Note that
$\distq_{s_i} (\sigset) \geq 1-x(\signal_i)$ iff $x(\signal_i)\geq 0$ (since player $i$ assigns probability $1$ to event $\sigset$).  Therefore,
$$R_i^1(\distq)=B_i^{1-x}(\sigset|\distq).$$
The action invest is rational for both players when $\sigset^2$ is $(1-x,1-x)-$ believed. We write this set of types as $B_*^{1-x,1-x}(\sigset^2|G)=B_*^{1-x,1-x}(\sigset^2|\distq)_1\times B_*^{1-x,1-x}(\sigset^2|\distq)_2$.

Let $R_i^2(\distq)$ be the set of types for which investment is level 2 rationalizable. This is the set of types $\signal_i\in B_*^{1-x,1-x}(\sigset^2|\distq)_i$ such that
$\distq_{s_i} (B_*^{1-x,1-x}(\sigset^2|\distq)_j) \geq 1-x(\signal_i)$.
Since player $j$ will not invest if $\signal_j\notin B_*^{1-x,1-x}(\sigset^2|\distq)_j$, the payoff from invest is at most $x(\signal_i)-(1-\distq_{s_i} (B_*^{1-x,1-x}(\sigset^2|\distq)_j))$. Therefore, if
$$
\distq_{\signal_i} (B_*^{1-x,1-x}(\sigset^2|\distq)_j)< 1-x(\signal_i),$$
player $i$ of type $\signal_i$ strictly prefers not invest over invest. Therefore,
$$R_i^2(\distq)=B_i^{1-x}(B_*^{1-x,1-x}(\sigset^2|\distq)).$$
Iterating this argument, we get
$$R_i^{k+1}(\distq)=B_i^{1-x}([B_*^{1-x,1-x}]^k(\sigset^2|\distq)).$$
The action invest is rationalizable for both players if it is $k-$th level rationalizable for all $k$. In other words, exactly when $\sigset^2$ is common $(1-x,1-x)-$ believed:
$$R^\infty(\distq)=C^{1-x,1-x}(\sigset^2|G).$$
By a symmetric argument, action not invest is rationalizable exactly when $\sigset^2$ is common $(x,x)-$ believed.

Note that for any information structure $\distq$, $[B_*^{1-x,1-x}]^k(\sigset^2|\distq)_1=[B_*^{1-x,1-x}]^k(\sigset^2|\distq)_2$ for any $k=1,2,\ldots \infty$. By definition, $F\cad G$ implies
for any $\signal_i\in [B_*^{1-x,1-x}]^k(\sigset^2|\distq)_i$,
$$\dist_{\signal_i} ([B_*^{1-x,1-x}]^k(\sigset^2|\distq)_j)\geq \distq_{\signal_i} ([B_*^{1-x,1-x}]^k(\sigset^2|\distq)_j)$$

Therefore,
$$B_*^{1-x,1-x}(\sigset^2|\distq)= B_*^{1-x,1-x}(\sigset^2|\dist)$$
$$B_i^{1-x}(B_*^{1-x,1-x}(\sigset^2|\distq))\subseteq B_i^{1-x}(B_*^{1-x,1-x}(\sigset^2|\dist))$$
$$\ldots$$
$$B_i^{1-x}([B_*^{1-x,1-x}]^k(\sigset^2|\distq))\subseteq B_i^{1-x}([B_*^{1-x,1-x}]^k(\sigset^2|\dist))$$
$$\implies \bigcap_{k\geq 1} [B_*^{1-x,1-x}]^k(\sigset^2|\distq) \subseteq  \bigcap_{k\geq 1} [B_*^{1-x,1-x}]^k(\sigset^2|\dist)$$
$$\text{or,} \ C^{1-x,1-x}(\sigset^2|\distq) \subseteq  C^{1-x,1-x}(\sigset^2|\dist)$$
$$\text{or,} \ R^\infty (\distq) \subseteq R^\infty (\dist).$$
This means if the action invest is rationalizable for a player $i$ of type $\signal_i$ under information structure $\distq$, then it remains rationalizable under $\dist\cad \distq$. Analogous arguments can be made for the action ``not invest'' as well.
\end{proof}

\subsection{Two-player Bank Run Example}
\label{sec:example}
Two agents each decide whether to stay (S) or run (R) on a bank. The payoff from running is $0$. The payoff from staying is $\theta$ if the other player stays, and $\theta-1$ if the other runs, where $\theta$ is the underlying unknown fundamental state~$\theta \in \{\frac{3}{2},\frac{1}{2}, -\frac{1}{2}\}$. The payoff matrix is as follows.
\begin{center}
\begin{tabular}{|c|c|c|}
\hline
& Stay & Run \\
\hline
Stay & $\theta, \theta$ & $\theta-1, 0$ \\
\hline
Run & $0, \theta-1$ & $0, 0$\\
\hline
\end{tabular}
\end{center}
Note that if $\theta=\frac{3}{2}$ then regardless of the other player's action, a player gets a higher payoff from staying. On the other hand, if $\theta=-\frac{1}{2}$ then regardless of the other player's action, a player gets a higher payoff from running.

In the introduction, we show that in this example, increasing correlation of signals by introducing positive $\alpha$ (see Figure \ref{fig:example}) can eliminate the maximal run equilibrium. Below, we assume a parametric information structure to illustrate how the set of equilibria changes as we increase $\alpha$.

To ensure that there is a range of $\alpha$ for which the minimal run equilibrium exists but the maximal run equilibrium does not exists, we need to introduce some asymmetry either in the prior or the information structure. We do this for the prior while keep the information structure symmetric. The state $\theta$ is drawn from the prior: $P(\theta=-\frac{1}{2})=\epsilon$, $P(\theta=\frac{3}{2})=2\epsilon$, and $P(\theta=\frac{1}{2})=1-3\epsilon$. Each player $i$ sees a signal $s_i\in\{\frac{3}{2},\frac{1}{2},-\frac{1}{2}\}$ drawn from a marginal distribution $P(s_i=\theta|\theta)=p$ and $P(s_i=\theta'|\theta)=\frac{1-p}{2}$ for $\theta'\neq \theta$.
Consider the following change in the signal structure, which makes the two signals more correlated. Suppose that the joint distribution of signals for state $\theta=\frac{1}{2}$ is as follows.
\[
\begin{array}{c|ccc|c}
\text{marg} & \frac{1-p}{2}  & p & \frac{1-p}{2}  & 1 \\ \hline
\frac{3}{2} & (\frac{1-p}{2})^2 -\color{red}{\alpha} & p(\frac{1-p}{2})+\color{blue}{\alpha} & (\frac{1-p}{2})^2 & \frac{1-p}{2} \\
\frac{1}{2} & p(\frac{1-p}{2}) +\color{blue}{\alpha} & p^2- \color{red}{2 \alpha}  & p(\frac{1-p}{2}) +\color{blue}{\alpha} & p \\
-\frac{1}{2} & (\frac{1-p}{2})^2 & p(\frac{1-p}{2}) +\color{blue}{\alpha} & (\frac{1-p}{2})^2 -\color{red}{\alpha} & \frac{1-p}{2}  \\
\hline
 s_1/ s_2 & -\frac{1}{2} & \frac{1}{2} & \frac{3}{2} & \text{marg} \\
\end{array}
\]
Note that when the state is $\theta=\frac{1}{2}$, the two signals are more likely to be high together or low together. Therefore, this change also means the signals are more interdependent according to supermodular or PQD order. $\alpha$ is a non-negative constant. If $\alpha=0$, then the signals are conditionally independent. A higher $\alpha$ makes the signals more interdependent. For simplicity, let us assume the joint distributions for state $\theta=\frac{3}{2}$ and $\theta=-\frac{1}{2}$ are conditionally independent.

We assume that $p>\frac{3-3\epsilon}{3-\epsilon}$, which ensures that when players see extreme signals, the updated beliefs are sufficiently high such that $R$ is strictly dominated for an agent with signal $s_i=\frac{3}{2}$ and $S$ is strictly dominated for an agent with signal $s_i=-\frac{1}{2}$.\footnote{An agent $i$ with signal $s_i$ believes that $P(\theta=\frac{3}{2}|s_i=\frac{3}{2})>P(\theta=-\frac{1}{2}|s_i=-\frac{1}{2})=\frac{p\epsilon}{p\epsilon+\frac{1-p}{2}(1-\epsilon)}>\frac{3}{4}$ when $p>\frac{3-3\epsilon}{3-\epsilon}$, which ensures the above dominated strategies.} The agent who receives a signal $s_i=\frac{1}{2}$ may or may not stay depending on her belief about the other player. Given an info structure $(\epsilon, p, \alpha)$, there may exist two symmetric BNE:
\begin{enumerate}
    \item Good Eqm: S if $s_i=\frac{3}{2}, \frac{1}{2}$ and R if $s_i=-\frac{1}{2}$ and
    \item Bad Eqm: S if $s_i=\frac{3}{2}$ and R if $s_i=\frac{1}{2}, -\frac{1}{2}$
\end{enumerate}
Let $e$ be a symmetric BNE, $e_G$ be the good eqm and $e_B$ be the bad eqm, and $\mathcal E$ be the set of symmetric BNE.

For the bad equilibrium $e_B$ to exist, it must be that the agent who sees signal $\frac{1}{2}$ wants to run if he believes that the other player plays the $e_B$ strategy. Note that for any state $\theta$, if an agent believes that the other player will stay with probability $q$, then her expected payoff from staying is $\theta q+(\theta-1)(1-q)=\theta+q-1$.
Therefore, $e_B$ exists if
$$\sum_\theta P \left(\theta|s_i=\frac{1}{2}\right) \left(\theta+P\left(s_j\in \{\frac{3}{2}\}|s_i=\frac{1}{2},\theta\right)-1\right)\leq 0$$
$$\mathbbm E \left[\theta|s_i=\frac{1}{2}\right]+\sum_\theta P \left(\theta|s_i=\frac{1}{2}\right) P\left(s_j\in \{\frac{3}{2}\}|s_i=\frac{1}{2},\theta\right) -1\leq 0.$$
Note that the LHS increases in $\alpha$. There exists $\alpha^*$ such that $e_B$ is an equilibrium iff $\alpha\leq \alpha^{*}$.

Similarly, for $e_G$ to exist, it must be that the agent who sees signal $\frac{1}{2}$ wants to stay if he believes that the other player plays the $e_G$ strategy.
One can show that there exists $\alpha^{**}>\alpha^*$ such that both the good and the bad equilibria $e_B,e_G$ exists for $\alpha\leq \alpha^*$ and only the good equilibrium $e_G$ exists for $\alpha\in[\alpha^*,\alpha^{**}]$.\footnote{$e_G$ exists if
$\mathbbm E \left[\theta|s_i=\frac{1}{2}\right]+\sum_\theta P \left(\theta|s_i=\frac{1}{2}\right) P\left(s_j\in \{\frac{3}{2},\frac{1}{2}\}|s_i=\frac{1}{2},\theta\right) -1\geq 0.$
Since the LHS decreases in $\alpha$, there is $\alpha^{**}$ such that $e_G$ exists for $\alpha\leq \alpha^{**}$. One can check that
$\alpha^{*}=-\left[\mathbbm E \left[\theta|s_i=\frac{1}{2}\right]+ P\left(s_j=\frac{3}{2}|s_i=\frac{1}{2}\right) -1\right] \frac{p}{P\left[\theta=\frac{1}{2}|s_i=\frac{1}{2}\right]}$
and
$\alpha^{**}=\left[\mathbbm E \left[\theta|s_i=\frac{1}{2}\right]- P\left(s_j=-\frac{1}{2}|s_i=\frac{1}{2}\right) \right] \frac{p}{P\left[\theta=\frac{1}{2}|s_i=\frac{1}{2}\right]}$
where the probabilities are based on the joint distribution before the change (conditionally independent). Finally, because apriori,  $\state$ is more likely to be $\frac{3}{2}$ than $-\frac{1}{2}$, $\mathbbm E \left[\theta|s_i=\frac{1}{2}\right]>\frac{1}{2}$ and $P\left(s_j=\frac{3}{2}|s_i=\frac{1}{2}\right)>P\left(s_j=-\frac{1}{2}|s_i=\frac{1}{2}\right)$, which ensures $\alpha^{**}>\alpha^*$.}

\begin{figure}
    \centering
\begin{tikzpicture}[x=10cm,y=1cm]

  \draw[thick] (0,0) -- (1,0);

  \fill (1/3,0) circle (1.2pt);
  \draw (1/3,0.08) -- (1/3,-0.08) node[below=2pt] {$\alpha^{*}$};

  \fill (1,0) circle (1.2pt);
  \draw (1,0.08) -- (1,-0.08) node[below=2pt] {$\alpha^{**}$};

  \node[above=6pt] at (1/6,0) {$e_G, \ e_B$};
  \node[above=6pt] at (2/3,0) {$e_G$};

\end{tikzpicture}
\caption{Correlation and Equilibrium}
\label{fig:CorrEqm}
\end{figure}
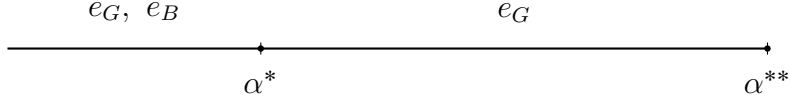

Consider a policymaker who does not want the players to run on the bank. Let $\mathcal R(e)$ be the expected number of players who run in equilibrium $e$.

Suppose that the policymaker anticipates that the player will play the worst equilibrium. Let us define the maximal expected run as
$$\mathcal R:=\max_{e\in \mathcal E} \mathcal R(e).$$
Note that the maximal expected run only depends on the marginal distribution of signal and not the joint distribution.\footnote{$\mathcal R(e_G)=2\times P\left(s_i=-\frac{1}{2},s_j=-\frac{1}{2}\right)+2\times 1\times P\left(s_i=-\frac{1}{2},s_j\neq -\frac{1}{2}\right)=2\times P\left(s_i=-\frac{1}{2}\right).$
Similarly, $\mathcal R(e_B)=2\times P\left(s_i\in\{-\frac{1}{2},\frac{1}{2}\}\right).$} Therefore, the impact of $\alpha$ on maximal expected run is only through the equilibrium. Recall that the bad equilibrium $e_B$ does not exist when $\alpha$ becomes sufficiently high ($\alpha>\alpha^*$). This means the maximal expected run $\mathcal R$ falls when $\alpha$ crosses the threshold $\alpha^*$. Thus, surprisingly, in a canonical model, increasing correlation of information decreases the maximal expected number of people that would run on the bank!

\end{document}